\documentclass[11pt,reqno]{amsart}

\usepackage[utf8]{inputenc}
\usepackage{amsfonts}
\usepackage{amsmath}
\usepackage{amssymb}
\usepackage{amstext}
\usepackage{amsthm}
\usepackage[tmargin=2cm,bmargin=2cm,lmargin=2cm,rmargin=2cm]{geometry}
\usepackage{comment}
\usepackage{mathtools}
\usepackage[colorlinks,citecolor=blue,urlcolor=blue]{hyperref}

\usepackage{multirow}

\newcommand{\p}{\mathbb{P}}

\newcommand{\e}{\mathbb{E}}
\newcommand{\D}{\mathrm{d}}

\newcommand{\bs}{\boldsymbol}

\renewcommand{\a}{{\alpha}}

\newcommand{\bt}{{\bs t}}
\newcommand{\ba}{{\bs\a}}

\newcommand{\R}{\mathbb{R}}

\theoremstyle{plain}
\newtheorem{theorem}{Theorem}
\newtheorem{lemma}[theorem]{Lemma}
\newtheorem{proposition}[theorem]{Proposition}
\newtheorem{corollary}[theorem]{Corollary}
\newtheorem{example}[theorem]{Example}

\theoremstyle{remark}
\newtheorem{remark}[theorem]{Remark}

\begin{document}

\author[J. Ivanovs]{Jevgenijs Ivanovs$^*$}
\thanks{$^*$Department of Mathematics, Aarhus University,
Ny Munkegade 118, DK-8000 Aarhus C, Denmark. Email: jevgenijs.ivanovs@math.au.dk. Supported by Sapere Aude Starting Grant 8049-00021B
`Distributional Robustness in Assessment of Extreme Risk'.}

\author[K. Yamazaki]{Kazutoshi Yamazaki$^{\dagger}$}
\thanks{$^{\dagger}$School of 
Mathematics and Physics, The University of Queensland, St Lucia,
Brisbane, QLD 4072, Australia. Email: k.yamazaki@uq.edu.au. In part supported by MEXT KAKENHI grant no.\ 19H01791 and 20K03758 and JSPS Open Partnership Joint Research
Projects grant no.\ JPJSBP120209921.  }

\date{}

\title[A series expansion formula of the scale matrix]{{
A series expansion formula of the scale matrix with applications in change-point detection}}

\begin{abstract}
%We introduce a new method for  the analysis of the CUSUM procedure in sequential change-point detection. 
% In the setting where observations are phase-type distributed with its post-change distribution given by exponential tilting of its pre-change distribution, the first passage analysis of the CUSUM statistic reduces to that of a certain Markov additive process. By utilizing the theory of the so-called scale matrix and further developing it, we derive exact expressions of the average run length, average detection delay, and false alarm probability under the CUSUM procedure.  The proposed method is robust and applicable in a very general setting with non-i.i.d.\ observations.  Numerical results are also given.
 
 We introduce a new  L\'evy fluctuation theoretic method to analyze the cumulative sum (CUSUM) procedure in sequential change-point detection. When observations are phase-type distributed and the post-change distribution is given by exponential tilting of its pre-change distribution, the first passage analysis of the CUSUM statistic is reduced to that of a certain Markov additive process. We develop a  novel series expansion formula of the scale matrix for Markov additive processes of finite activity, and apply it to  derive exact expressions of the average run length, average detection delay, and false alarm probability under the CUSUM procedure. %The proposed method is robust and applicable in a general setting with non-i.i.d.\ observations. Numerical results also are given.

%Given a sequence of phase-type distributed observations and its distribution changes by exponential tilting, 
%
%In the setting 
%
%Under the setting observations are phase-type distributed 
%
% under phase-type distributed observations. 
%
%via  the fluctuation theory of L\'evy and Markov additive processes.
%
%
% when the observation is phase-type distributed. 
%
%Using the 
%
%
%Motivated by the recent progress of the fluctuation theory of L\'evy and Markov additive processes, 
%
%new explicit expression of the scale function.
%
%non-i.i.d. observation.

\noindent \small{\noindent  {AMS 2020} Subject Classifications: 
60G51, %levy 
%62L10,
60G40,
%62L15,
62M05.
%94C12.
%93E20, %Optimal stochastic control [Is this relevant?]
%49J55 %Existence of optimal solutions to problems involving randomness  
%91A05 %2-person games
%91A15, %Stochastic games, stochastic differential games
%91A27, %Games with incomplete information, Bayesian games 
%91A35 %Decision theory for games   
%91A55 %Games of timing
%91B06 %Decision theory (in Math Econ) 
%90B06 %Transportation, logistics and supply chain management [Since our application has to do with supply chains]
%90B50. %Management decision making, including multiple objectives [This is to reflect the non-zero-sum feature]
\\
%JEL Classifications: \\
\textbf{Keywords:} L\'evy processes, Markov additive processes, scale matrices, phase-type distributions, hidden Markov models, CUSUM.}
\end{abstract}

\maketitle

\section{Introduction}

%
%Sequential change-point detection is a classical statistical decision problem where the aim is to detect changes of an unobservable system through indirect observations as quickly and as accurately as possible.  This finds applications essentially in all fields of engineering as well as natural and social sciences. 
%Classical applications of change-point detection include statistical quality control \cite{bissell1969cusum, lucas1982combined, lucas2000fast, woodall1985multivariate}, signal processing \cite{andre1988new, jones1970change}, seismology \cite{nikiforov1985application, popescu2011detection}, finance/economics \cite{shiryaev2002quickest}, and epidemiology \cite{baron2004early,yu2013change}. For the technological developments toward automation and unmanned operation, efficient detection schemes are becoming increasingly important.
%%Mathematical modelling for timely responses against viruses is essential in various fields.
%Cyber- and bio-security are emerging fields of research where mathematical modeling for efficient detection is essential for saving people's lives, intellectual properties, and the economy. We refer the reader to e.g.\ \cite{poor2013introduction,poor2008quickest,  Shiryaev1978, MR799155,tartakovsky2014sequential}  for books on change-point detection and related sequential analysis problems.

Sequential change-point detection is a classical sequential decision problem where the aim is to identify changes in an unobservable system through indirect observations quickly and accurately. This has applications in all fields of engineering as well as in natural and social sciences. Classical applications of change-point detection include quality control \cite{bissell1969cusum, lucas1982combined, woodall1985multivariate}, signal processing \cite{andre1988new, jones1970change}, seismology \cite{nikiforov1985application},
% popescu2011detection}, 
 finance/economics \cite{shiryaev2002quickest}, and epidemiology \cite{baron2004early}.
 %yu2013change}. 
 For the technological developments toward automation and unmanned operation, efficient detection schemes are becoming increasingly important. Cyber- and bio-security are emerging fields where mathematical modeling for efficient detection is essential for saving people’s lives, intellectual property, and the economy. We refer the reader to \cite{poor2013introduction,poor2008quickest, MR799155,tartakovsky2014sequential}  for books on change-point detection and related sequential analysis problems.

 The cumulative sum (CUSUM) procedure, originally developed by Page  \cite{page1954continuous}, is one of the most used detection rules. It is also theoretically important because of its optimality in the sense of minimizing the \emph{Lorden detection measure} in the \emph{minimax formulation}. Lorden \cite{lorden1971procedures}  first proved its asymptotic optimality, and Moustakides \cite{moustakides1986optimal}  showed its exact optimality. One major reason for its popularity is its implementability. Its alarm time is concisely given by the first passage time of the CUSUM statistic, which is, in the terms of probability theory, the \emph{reflected process} of the \emph{log-likelihood ratio (LLR) process}. Wald’s approximation and renewal theoretic methods are popular tools for its analysis. However, they usually lead to approximate or asymptotic results. Exact computation of the performance measures, such as the average run length, average detection delay, and false alarm probability, is rarely achieved. Contrary to the continuous-time model, where analytical tools such as martingale methods and It\^o  calculus are available (see, e.g., \cite{peskir2006optimal, peskir2000sequential}), exact computation of the first passage identities for discrete-time processes tends to be infeasible by traditional methods. For this reason, research on the CUSUM statistic, which is a discrete-time process, has focused on pursuing approximate results. We refer the reader to, e.g., \cite[Ch.\ 8]{tartakovsky2014sequential} for a detailed review and numerical methods for the CUSUM stopping rule.

% In sequential testing, it is very rare that the analytical results can be obtained explicitly.  The only analytical results possibly obtained for asymptotic results, that are derived by renewal theory arguments. 
%[See Gut and Siegmund. A sequence of asymptotic theory of sequential analysis are developed by, e.g., Lai and Tartakovsky.]

%In this paper, we obtain \emph{exact expressions} of the performance measures of the CUSUM procedure, focusing on the case where observations are phase-type (PH) distributed with its post-change distribution given by exponential tilting of its pre-change distribution.  A PH distribution is given as that of the absorption time of a finite-state continuous-time Markov chain consisting of transient states and a single absorbing state.
%Examples of PH distributions include (hyper-)exponential, Erlang and Coxian distributions. 
%The class of the PH distribution is known to be dense in the class of all positive distributions, and hence, in principle, any positive  distribution can be approximated by PH distributions. A wide array of fitting algorithms is provided, for example, in \cite{asmussen1996fitting,horvath2002phfit,  okamura2011refined, thummler2006novel}. In particular, when a distribution has a  completely monotone density, there exist fitting algorithms via hyper-exponential distributions such as \cite{feldmann1998fitting}, which are guaranteed to converge to the true distribution.  We refer the reader to \cite{APQ} and references therein for a comprehensive review of the PH distribution.

In this paper, we develop a new series expansion formula of the so-called scale matrix and apply it in obtaining
 \emph{exact expressions} of the performance measures of the CUSUM procedure, focusing on a case where observations are phase-type (PH) distributed and the post-change distribution is given by exponential tilting of its pre-change distribution. A PH distribution is given as the absorption time distribution of a finite-state continuous-time Markov chain comprising transient states and a single absorbing state. Examples of PH distributions include (hyper-)exponential, Erlang, and Coxian distributions. The class of PH distributions is dense in the class of all positive distributions in the sense of weak convergence. Hence, in principle, any positive distribution can be approximated by PH distributions. An array of fitting algorithms is provided, for example, in \cite{asmussen1996fitting,horvath2002phfit,  okamura2011refined, thummler2006novel}. In particular, when a distribution has a completely monotone density, there are fitting algorithms using hyper-exponential distributions such as  \cite{feldmann1998fitting}, which are guaranteed to converge to the true distribution. We refer the reader to \cite{APQ}  and references therein for a comprehensive review of the PH distribution, and recent results on matrix-exponential distributions such as \cite{bean2022rap,cheung2022multivariate}. 

The scale matrix is defined for 
%main tool we use is the \emph{fluctuation theory} of 
\emph{Markov additive processes (MAPs)}, whose research is receiving much attention but is still under development and rarely applied in the literature on sequential testing. %A MAP is a generalization of L\'evy processes modulated by a continuous-time Markov chain. 
A MAP is a bivariate Markov process $(X,J)$, where the increments of $X$, called the \emph{ordinator}, are governed by a continuous-time Markov chain $J$, called the \emph{modulator}. Conditionally on $J = i$, the ordinator $X$ evolves as some  L\'evy process, say $X^{(i)}$, until $J$ changes its state to some $j$, at which instant an independent jump specified by the pair $(i,j)$ is introduced into $X$.
% This process enjoys a so-called \emph{scale matrix}, by which many first-passage identities can be expressed analytically. 
 As obtained, for example in \cite{IP12}, many first-passage identities can be expressed in terms of the scale matrix. However, its applications have been severely limited due to the difficulty of computing the scale matrix. In this paper, we develop an exact and computationally feasible method of computing the scale matrix for a wide subset of MAPs, including those required for the analysis of CUSUM.

% \red{By the scale matrix, many first-passage identities can be expressed analytically. However existing results on the computational methods of the scale matrix are significantly limited. In this paper, we develop an exact and computationally feasible method of computing the scale matrix for a wide subset of MAPs.}

The connection between the LLR process and the MAP is established as follows.

 Given two probability distributions, $F_0$ and $F_1$, on $(0,\infty)$ with their respective densities $f_0$ and $f_1$,  the LLR process under a sequence of observations $\zeta = (\zeta_1, \zeta_2, \ldots)$ is given by
\begin{align}
L_n :=  \sum_{i=1}^n \log  \frac {f_1(\zeta_i)} {f_0(\zeta_i)}, \quad n \geq 1. \label{LLR}
\end{align}
When $\zeta$ are independent and identically distributed (i.i.d.), $L$ becomes a random walk. However, the observations $\zeta$ are, in general, not necessarily  $F_0$- or $F_1$-distributed.
When $f_1$ is given by exponential tilting of $f_0$, i.e.,
$f_1(x)=  {\exp(\theta x)f_0(x) } / \int_0^\infty \exp(\theta y) F_0(\D y )
% {\mathrm{E}_0 [\exp(\theta \zeta_1)]}
 $, $x > 0$, 
for some $\theta$,
% (where $\mathrm{E}_0 [\exp(\theta \zeta_1)]$ is the Laplace transform of $\zeta_1$ when it is $F_0$-distributed), %We assume $\lambda \neq 0$ because otherwise $f_1$ and $f_0$ are not distinguishable.
then the LLR process is reduces to 
\begin{align*}
L_n =  \sum_{i=1}^n (\theta  \zeta_i - \kappa(\theta) 
%\mathrm{E}_0 [e^{\theta \zeta }]
), \quad n \geq 1,
\end{align*}
where $\kappa(\theta) := \log \int_0^\infty \exp(\theta y) F_0(\D y )$.
%With the cumulant $\kappa(\theta) := \log \mathrm{E}_0 [e^{\theta \zeta }]$, the LLR function becomes
%\[\log \frac{f_1(x)}{f_0(x)} = \theta x- \log \mathrm{E}_0 [e^{\theta \zeta }]. \]
%Namely, the jump size of the LLR process is given by $\theta  \zeta_i - \log \mathrm{E}_0 [e^{\theta \zeta }]$. 
The idea  is to consider a continuous-time process, say $X = (X_t)_{t \geq 0}$, with a constant drift $\theta$ and jumps of constant size 
$\kappa(\theta)$
%$\log \mathrm{E}_0 [\exp(\theta \zeta )]$ 
with interarrival times $\zeta$ so that the $n$-th post-jump location of process $X$ coincides with $L_n$, i.e.,\ the LLR after the $n$-th observation. Furthermore, when $\zeta$ are PH distributed,  $X$ can be modeled as a MAP, using (a slight modification of) the Markov chain that describes the PH distribution as its modulator. Similar methods are used in \cite{asmussen2004russian} to analyze  PH L\'evy processes.

% or a generalization of L\'evy processes modulated by a continuous-time Markov chain. 

To the best of our knowledge, Albrecher et al.\ \cite{AAI14} is the only existing work that uses the fluctuation theory of MAPs in sequential analysis problems. In \cite{AAI14}, they focus on Wald's sequential probability ratio test (SPRT) in classical  binary sequential hypothesis testing, where observations $\zeta$ are $F_0$ or $F_1$-distributed. Hence, the LLR process $L$ becomes a random walk. In this simple i.i.d.\ setting,  the ordinator $X$ of the MAP, used instead of $L$, is reduced to 
%the analysis of the LLR process can be reduced to that of 
a \emph{Sparre-Andersen process}, a generalization of the compound Poisson process with jump times given by a renewal process. The authors obtain exact expressions of the expected sample size and Type I and Type II error probabilities written in terms of the scale matrix. 

This paper is concerned with the CUSUM statistic
\begin{align}
R_n := \max_{1 \leq k \leq n} \sum_{i=k}^n \log  \frac {f_1(\zeta_i)} {f_0(\zeta_i)}, \quad n \geq 1,  \label{CUSUM}
\end{align}
which is obtained by reflecting the LLR process $L$ at a lower boundary of zero.
The CUSUM procedure triggers an alarm at the first moment this reflected process up-crosses a fixed threshold. Contrary to the study of the classical SPRT  \cite{AAI14}, \emph{the observations $\zeta$ fail to be i.i.d}. Thus, novel and more flexible approaches are required to tackle this problem. 

% requiring us to develop novel and more flexible approaches.
%requiring a need of considering a generalization of the reflected Sparre-Andersen process.

To enjoy the theory of MAPs, a connection between the CUSUM statistic and the reflection of a MAP must be established first, like that between the LLR and the MAP described above. In other words, we need to construct a continuous-time process $X$ so that the post-jump points of its reflected path coincide with 
 the CUSUM statistic for general observations $\zeta$, which are not necessarily $F_0$ or $F_1$-distributed. Contrary to the SPRT case  \cite{AAI14}, symmetry is lost when considering reflection at a one-sided boundary. For this reason, we require two  approaches, depending on the sign of the tilting parameter $\theta$. When $\theta > 0$, we use $X$ with a positive drift and negative jumps; when $\theta < 0$, we use $X$ with a negative drift and positive jumps.  Because the discrete-time stochastic process $R$ is expressed by means of a continuous-time stochastic process $X$, careful pathwise analysis is required.  The starting point and the (alarm-triggering) barrier must be adjusted, depending on the sign of $\theta$ (see Figure \ref{fig_LLR_CUSUM}).

We then conduct the first-passage analysis of the constructed continuous-time process $X$ by writing it as a MAP, by suitably choosing the modulator $J$.  The biggest challenge is to attain analytical and explicit results without losing the generality of the law of the observation process $\zeta$ and the change point. To present our mathematical derivation efficiently, we take two steps.

\begin{enumerate}
\item We first consider the case the observations $\zeta$ are i.i.d. This is required for computing the optimal barrier in the minimax formulation. That is the barrier where the average run length, when $\zeta$ is independent and $F_0$-distributed, equals a given parameter. In this case, as in \cite{AAI14}, the process $X$ becomes an ordinary Sparre-Andersen process, represented as a MAP modulated by a modification of the Markov chain for the PH distribution of $\zeta$.  Using the first passage identity of its reflected process  given in \cite{IP12}, the moment generating function (and thus the first moment as well) of the average run length can be written in terms of the corresponding scale matrix, for which a new series expansion formula is derived.

\item We then extend it to the case with a change point. While it is not common to be pursued in non-Bayesian formulations because of its difficulty via traditional methods, we carry out an exact computation of the performance measures of CUSUM under a set of change-point distributions. In particular, we provide an exact computation of the average run length, average detection delay, and false alarm probability when the change-point is \emph{discrete-time PH distributed}. This enables us to study, for example, the cases of geometric, 
negative binomial, and mixed geometric distributions, which are special cases of the discrete-time PH distribution (see, e.g., \cite{neuts1994matrix}). To this end, we introduce another (this time, discrete-time) Markov chain, say $Z$, which changes its state after each observation. The change-point is modeled by the first time it enters a certain subset of its state space. 
The analysis of the CUSUM statistic in this general setting is made possible by generalizing the modulator $J$ so that it keeps track of the evolution of the Markov chain for $\zeta$ and for $Z$.
%as the state of the discrete-time PH distribution for the time-change.
\end{enumerate}

%In other words, this models the case the change point time is 

% while at the same time not losing  its analytical tractability.  

In fact, this can be further generalized by allowing the observation distribution to depend on the Markov chain $Z$. This is particularly important because, while for the design of the CUSUM rule two distributions $F_0$ and $F_1$ must be specified a priori, the true observations can be non-i.i.d.\ with distributions other than $F_0$ and $F_1$. By considering observation distributions perturbed from  $F_0$ and $F_1$,  the robustness of the designed CUSUM procedure can be evaluated analytically. This generalization can be seen as a type of \emph{hidden Markov model} of sequential change-point detection. 
%, where the evolution of the Markov chain is not directly observable and the observation distribution depends on the state of the Markov chain. 
See \cite{Dayanik2009,  Dayanik2021, fuh2003sprt,fuh2004asymptotic,fuh2018asymptotic} for various hidden Markov models.
It is noted that the CUSUM procedure in our Markov-modulated generalization is still defined by the simple LLR function between $f_0$ and $f_1$ (with general non-i.i.d.\ observations $\zeta$). Many non-i.i.d.\ models consider more complex LLR functions by assuming to know the conditional density given past observations (i.e.,\ $f_i(\cdot | \mathcal{X})$ for every history $\mathcal{X}$ of observations for $i = 0,1$), and they pursue only approximate results. However, in practice, conditional densities are often too complex to be calibrated, and  simpler rules are often preferred.
%it is more common for practitioners to prefer a simple rule.  
Here, we stick to this simple form of the LLR function and obtain exact and concrete results. 

The above procedures reduce the computation for the CUSUM procedure to that of the scale matrix. 
Hence, it is essential to develop a way to compute the scale matrix to carry out the introduced techniques in practice.  \emph{This is an important component missing in \cite{AAI14},} where the expression of the scale matrix was obtained only for the case $\zeta$ are Erlang distributed. In this paper, we derive a new series expansion formula for the scale matrix for MAPs with a constant drift and general finite-activity one-sided jumps (Theorem \ref{thm:W}).  This generalizes the series expansion formula of the scale function obtained in \cite[Thm.\ 2.2]{LW20} for the Cram\'er-Lundberg process.   In particular, the scale matrix required for the above computation for CUSUM can be explicitly and analytically written as a sum of matrix exponentials. This enables us to conduct exact computations of the performance measures for the CUSUM procedure, even for the general case modulated by $Z$. 

Our series expansion formula of the scale matrix is important in its own right. While the research on MAPs and the scale matrix is relatively new, many quantities of interest are already known to be expressible by means of the scale matrix (see, e.g., \cite{breuer2008first,ccauglar2021optimal, feng2014potential, IP12,kyprianou2008fluctuations}). This is analogous to how the scale function is used for the ordinary L\'evy process (see, e.g., \cite{bertoin1996levy, kuznetsov2012theory, kyprianou2014fluctuations}). 
However, different from the scale function that can be computed by straightforward Laplace inversion, the computation of the scale matrix is challenging and thus it has been a major obstacle to its practical applications. This new formula derived in this paper can be directly used for the study of SPRT \cite{AAI14} and its generalizations.  In addition, it has a direct contribution to the study of the Sparre-Andersen process as in \cite{borovkov2008ruin, cheung2010structural,  cheung2011orderings,gerber2005time}. These, besides the applications in sequential analysis as in this paper,  have broad applications in, e.g., insurance mathematics and queueing analysis.

To confirm the analytical results and computational feasibility, we conduct numerical experiments. We consider both simple and complex cases with non-i.i.d.\ observations. We test the results using the introduced scale matrix approach against those approximated by Monte Carlo simulation, confirming the accuracy and efficiency of the proposed method.
% that the results are indeed accurate and computation can be done efficiently.

The rest of the paper is organized as follows. In Section \ref{section_setting}, we review the CUSUM procedure and construct a continuous-time process whose reflected path coincides with 
 the CUSUM statistic. In Section \ref{sec_fluctuations_SA}, we consider the case where $\zeta$ are i.i.d. We review  the fluctuation theory and the scale matrix, and write the average run length in terms of the scale matrix.  In Section \ref{section_series_expansion},  we derive our series expansion of the scale matrix,  with which the scale matrix required for the analysis of the CUSUM statistic is written explicitly. In Section \ref{section_non_iid}, we generalize the results by introducing a discrete-time Markov chain, and obtain the average run length, average detection delay, and false alarm probability for non-i.i.d.\ cases. We conclude the paper with numerical results in Section \ref{sec_numerics}. Some proofs are deferred to the appendix.

\section{Preliminaries} \label{section_setting}

In this section, we review the classical sequential change-point detection problem and the CUSUM procedure.  We then focus on the case the post-change distribution is given by exponential tilting of the pre-change distribution, and construct a generalization of the Sparre-Andersen process so that the post-jump points of its reflected path coincide with the CUSUM statistic.

\subsection{Change-point detection and CUSUM} \label{subsection_change_point}
To describe the classical CUSUM procedure, we first consider the classical setting where the observations are i.i.d.\ before and after the change, conditionally given the change point.   

Suppose a sequence $\zeta = (\zeta_1, \zeta_2, \ldots)$ of independent random variables are observed sequentially. At an unobservable disorder time $\nu \geq 0$, it changes its distribution from $F_0$ to $F_1$. In other words, conditionally given $\nu$, random variables $\zeta_i$ and $\zeta_j$ are independent for $i \neq j$ and $\zeta_1, \zeta_2, \ldots, \zeta_{\nu} \sim F_0$ and $\zeta_{\nu+1}, \zeta_{\nu+2}, \ldots \sim F_1$ (we follow the convention that the disorder is triggered \emph{immediately after} $\nu$-th observation). Here, we allow $\nu$ to take zero with a positive probability; in this scenario, the observation $\zeta$ is $F_1$-distributed from the first observation. We also allow $\nu$ to be infinity (and hence $\nu$ is $F_0$-distributed at all times) with a positive probability.

% respectively so that
%%For an unobservable change time $\theta \geq 0$, suppose
%\[
%\underbrace{\zeta_1, \zeta_2, \ldots, \zeta_{\theta}}_{\sim f_0}, \underbrace{\zeta_{\theta+1}, \zeta_{\theta+1}, \ldots }_{\sim f_1}
%\]
%\kazu{use $\zeta$ or $\xi$? I am ok with both.}

The objective of sequential change-point detection is to identify the disorder as quickly as possible and as accurately as possible. A strategy 
 $T$ is selected from the set $\mathcal{T}$ consisting of all stopping times with respect to the filtration generated by the observation $\zeta$, namely $\mathbb{F} := (\mathcal{F}_n)_{n \geq 0}$ with $\mathcal{F}_0$ the trivial $\sigma$-algebra and $\mathcal{F}_n := \sigma(\zeta_1, \ldots, \zeta_n)$ for $n \geq 1$. For each constant $k \geq 0$, let $\mathrm{P}_k$ be the conditional probability under which $\nu = k$ and $\mathrm{E}_k$ the corresponding expectation. In particular, under $\mathrm{P}_0$ (resp.\ $\mathrm{P}_\infty$), $\zeta$ are independent and $F_1$- (resp.\ $F_0$-)distributed. As is commonly assumed in the literature, let $F_0$ and $F_1$ admit densities  $f_0$ and $f_1$, respectively, with respect to some baseline measure.

%This paper focuses on the (Page's) CUSUM strategy.  
The  CUSUM statistic after $n$ observations is given by \eqref{CUSUM}.
%\jev{should we just refer to p. 3?}
%\begin{align}
%R_n := \max_{1 \leq k \leq n} \sum_{i=k}^n \log  \frac {f_1(\zeta_i)} {f_0(\zeta_i)}, \quad n \geq 1. \label{CUSUM}
%\end{align}
For convenience sake, we also let  $R_0 := 0$.
The CUSUM procedure, parameterized by a constant $A > 0$, triggers an alarm at the first time \eqref{CUSUM} exceeds $A$, namely
\begin{align}
T_A := \inf \{ n \geq 1: R_n > A \}. \label{T_A}
\end{align}
It is known that \eqref{CUSUM} admits a recursive relation (see, e.g., \cite[Eq.\ (2.6)]{moustakides2009numerical}):
\[
R_n = \max \Big(0, R_{n-1} +  \log  \frac {f_1(\zeta_n)} {f_0(\zeta_n)} \Big), \quad n \geq 1.
\]
With the LLR process defined in \eqref{LLR} and its (capped) running minimum process
\begin{align*}
%L_n :=  \sum_{i=1}^n \log  \frac {f_1(\zeta_i)} {f_0(\zeta_i)} \quad \textrm{and} \quad 
\underline L_n := \Big(\min_{1 \leq k \leq n} L_k \Big) \wedge 0, \quad n \geq 0,
\end{align*}
we can write
\begin{align}
R_n=L_n-\underline L_n, \quad n \geq 0. \label{R_n}
\end{align}
In particular, under $\mathrm{P}_\infty$ and $\mathrm{P}_0$ with i.i.d.\ observations, $L$ reduces to a random walk and $R$ is its reflected process.
%
%
%%In this formulation, the CUSUM procedure is optimal. \red{[I will check sufficient conditions for the optimality.]} 
%Let
%\[
%V_n := \max_{1 \leq k \leq n} \prod_{i=k}^n \frac {f_1(\zeta_i)} {f_0(\zeta_i)}.
%\]
%CUSUM stops at its first up-crossing time
%\[
%T_A := \inf \{ n \geq 1: V_n \geq e^A \}.
%\]
%for $A > 0$ to be selected.
%We can also write
%\[
%T_A := \inf \{ n \geq 1: R_n \geq A \},
%\]
%where $R$ is CUSUM statistic defined by $R_n = 0$ and 
%\[
%R_n := \log V_n = \max_{1 \leq k \leq n} \sum_{i=k}^n \log  \frac {f_1(\zeta_i)} {f_0(\zeta_i)}, \quad n \geq 1.
%\]
%Note that 
%\[
%R_n = \max \Big(0, R_{n-1} +  \log  \frac {f_1(\zeta_n)} {f_0(\zeta_n)} \Big).
%\]

The CUSUM procedure is well-known for its optimality properties in the minimax formulation.  From \cite{moustakides1986optimal}, it is known that, given a parameter $\beta \geq 1$ selected by the decision maker, the CUSUM procedure \eqref{T_A} with the selection of the barrier $A_\beta$  satisfying
\begin{align}
 \mathrm{E}_\infty (T_{A_\beta}) = \beta \label{lorden_beta}
\end{align}
 is optimal in the sense that it minimizes the Lorden detection measure \cite{lorden1971procedures}:
\begin{align}
C(T) := \sup_{k \geq 0} \mathrm{ess} \sup \mathrm{E}_k \big((T-k)^+ | \mathcal{F}_k \big) \label{lorden_measure}
\end{align}
over the set of strategies
\[
\Delta_\beta  = \{ T \in \mathcal{T}: \mathrm{E}_\infty (T) \geq \beta \}.
\]
%for selected $\beta \geq 1$. %Historically,  Lorden \cite{lorden1971procedures} first showed the asymptotic optimality as $\beta \to \infty$ and later Moustakides \cite{moustakides1986optimal}  showed the exact optimality.

The Lorden detection measure \eqref{lorden_measure} only evaluates the worst-case performance and is often not suitable in real applications. It is thus important to consider other measures as well to evaluate a detection strategy. Popular measures, often used in Bayesian formulations,  are the average run length, average detection delay and false alarm probability, respectively given by: \begin{align}
\mathrm{ARL}(T) &:= \mathrm{E} (T), \label{def_ARL} \\
\mathrm{ADD}(T) &:= \mathrm{E} \big((T-\nu)^+\big),  \label{def_ADD} \\
\mathrm{PFA}(T) &:= \mathrm{P} (T \leq \nu).  \label{def_PFA}
\end{align}
For the above probability and expectations to make sense, the law $\mathrm{P}$ of the change point $\nu$ and the observation process $\zeta$ must be completely specified.  For example, for the computation of the optimal barrier satisfying \eqref{lorden_beta}, $\mathrm{ARL}(T_A) =  \mathrm{E}_\infty (T_A)$ with $\mathrm{P} = \mathrm{P}_\infty$. It is also of interest to consider the case of  $\mathrm{P} = \mathrm{P}_k$ where $\nu = k$ a.s.\ for $0 \leq k < \infty$.

\subsection{Our assumption}

%For the rest of the paper, we assume that $F_0$ is a phase-type distribution and $F_1$ its exponential tilting. More specifically,  we assume $F$ is phase-type with representation $()$

We assume both $F_0$ and $F_1$ (that define the LLR and CUSUM statistic) are positive distributions (with support $(0,\infty)$) and their densities satisfy
%Let $f_0(x)=\ba e^{\bs T x}\bs t$ be a phase-type density, and let 
\[
f_1(x)= \frac {e^{\theta x}f_0(x) } {
 \int_0^\infty e^{\theta y} F_0(\D y) 
%\mathrm{E}_0 [e^{\theta \zeta_1 }
}, \quad x > 0,
\]
%for
%\[
%G_0(\lambda) := \e_0 [e^{\lambda \zeta }]
%\]
for some known parameter $\theta\in (-\infty,\overline\theta) \backslash \{0\}$ where $\overline{\theta} := 
\sup \{ \theta \geq 0: \int_0^\infty e^{\theta y} F_0(\D y) 
% \mathrm{E}_0 [e^{\theta \zeta_1 }]
< \infty \}$. %We assume $\lambda \neq 0$ because otherwise $f_1$ and $f_0$ are not distinguishable.
With the cumulant 
\[
\kappa(\theta) := \log  \Big( \int_0^\infty e^{\theta y} F_0(\D y) \Big), 
%\mathrm{E}_0 [e^{\theta \zeta_1 }],
\]
 the LLR function becomes
\[\log \frac{f_1(x)}{f_0(x)} = \theta x-\kappa(\theta), \quad x > 0. \]
 In particular, under $\mathrm{P}_0$ and $\mathrm{P}_\infty$,
the LLR process $L$  as in \eqref{LLR} becomes a random walk with i.i.d.\ increments $(\theta \zeta_n -\kappa(\theta))_{n \geq 1}$.

Well-known examples satisfying this exponential tilting assumption are two exponential densities and two Erlang densities with fixed shape parameter. As is shown in \cite{asmussen1989exponential}, an exponentially tilted distribution of the PH distribution is again PH (see \cite[Eq.\ (10)]{AAI14} for the formula).  Explicit results for the classical SPRT were obtained for the exponential case in \cite{TvA86} and the Erlang case in \cite{AAI14}. However, beyond these results, exact expressions of the performance measures are rarely obtained  in sequential analysis, even with the assumption of exponential tilting.

\subsection{Alternative expression of the CUSUM statistic}  We shall now express the CUSUM statistic in terms of a reflected path of a certain continuous-time process.
 While the definition of the CUSUM  procedure \eqref{CUSUM} requires the densities $f_0$ and $f_1$ to be specified, the process  \eqref{CUSUM} is well-defined even when $\zeta$ are non-i.i.d.\ with distributions other than $F_0$ and $F_1$.  In the subsequent discussions, let   $\zeta$ be \emph{any strictly positive sequence}.
 % which do not need to be $F_0$ nor $F_1$-distributed.
 
% In real applications, $f_1$ and $f_2$
% 
%  (especially the latter) are difficult to be calibrated a prior and, what is worse, the observations may fail to be i.i.d. Despite these, CUSUM is a popular method due to its simplicity and $f_1$ and $f_2$ are selected to be some suitable ones even when the above assumptions fail. It is therefore important to evaluate its performance for a general observation process $\zeta$. This is done in Section  \ref{section_non_iid}.
% 
% To spell this out, we use a different letter for the observation process.
% Below, let $\xi$ to be any stochastic process, which do not need to be $F_0$ nor $F_1$-distributed.  It reduces to the classical setting when $\xi \equiv \zeta$.

Let $N = (N_t)_{t \geq 0}$ be a counting process with $N_0 = 0$ with its $l$-th jump time given by the sum of the first $l$ observations
\[
\eta_l =\zeta_1 + \zeta_2 + \cdots + \zeta_l, \quad l \geq 0.
\]
%\begin{remark} Here, we do not require $\zeta$ to be i.i.d.\ as long as $\zeta_k > 0$ for all $k \geq 1$. \kazu{to be checked if this does not cause any problem. does not need to be independent either.} 
%\end{remark}
We then introduce a continuous-time process
\begin{align}
X_t^{(x)}= x + \theta t-\kappa(\theta) N_t, \quad t \geq 0, \label{Sparre-andersen_CUSUM}
\end{align}
started at $x \in \R$.
%Here, $N$ is a generalization of the renewal process with $F_0$-distributed interarrival times until $\nu$-th arrival and $F_1$-distributed interarrival times thereafter. 
%\red{[TODO: this restriction can later be relaxed. can be any counting process?]} 
In particular, when $\zeta$ are i.i.d., $N$ reduces to  an ordinary renewal process and hence the process $X$ falls in the class of what is called Sparre-Andersen processes in actuarial science. We refer the readers to, e.g., \cite{borovkov2008ruin, cheung2010structural,  cheung2011orderings,gerber2005time} for existing research on the Sparre-Andersen process. For the rest of the paper, let us call \eqref{Sparre-andersen_CUSUM} a \emph{generalized Sparre-Andersen process} to include the cases $\zeta$ are non-i.i.d.

%
%\red{TODO}
%Also define jump times of  the renewal process $N$ (or $X$) by

Our key observation is the equivalence of the first passage time of the CUSUM statistic \eqref{CUSUM} and that of the reflected process of $X^{(x)}$ defined by 
\begin{align}
Y_t^{(x)} := X_t^{(x)} - \underline{X}_t^{(x)} \wedge 0, \quad t \geq 0, \label{reflected_process}
\end{align}
where $\underline{X}_t^{(x)} := \inf_{0 \leq s \leq t} X_s^{(x)}$.
We denote the first passage time of \eqref{reflected_process} by
\begin{align}
\tau_a^{(x)} :=\inf\{t\geq 0: Y_t^{(x)}> a\},\qquad  a > 0. \label{tau_B}
\end{align}
For simplicity, we drop the superscript when $x = 0$ and write $X = X^{(0)}$, $Y = Y^{(0)}$, and $\tau_a = \tau_a^{(0)}$. 

By construction, we have $\log (f_1(\zeta_l) / f_0(\zeta_l)) = \theta \zeta_l-\kappa(\theta) = X_{\eta_l} - X_{\eta_{l-1}} $ for $l \geq 1$, and hence %it is clear that 
%the laws of  $\{ X_{\eta_1}, X_{\eta_2}, \ldots \}$ and $\{ S_1, S_2, \ldots \}$ are the same.  I do not know if it is precise but let me write for now 
\begin{align}
L_n = X_{\eta_n} \quad \textrm{and} \quad \underline L_n = \min_{1 \leq k \leq n} X_{\eta_k} \wedge 0, \quad n \geq 0. \label{connection_S_X}
\end{align}

We deal with the cases $\theta$ is positive and negative separately because the behavior of $X^{(x)}$ differs depending on the sign of $\theta$ as in the following remark.
\begin{remark} \label{remark_sign}
Because 
\[
\mathrm{sgn} (\theta) = \mathrm{sgn} (\kappa(\theta)),
\]
%the sign of the cumulant $\kappa(\lambda)$ coincides with the sign of~$\lambda$, 
%we have
\begin{enumerate}
\item when $\theta > 0$, $X^{(x)}$ has a constant positive drift with negative jumps;
\item when $\theta < 0$, $X^{(x)}$ has a constant negative drift with positive jumps.
\end{enumerate}
Following the terminology of the theory of L\'evy processes, we call $X^{(x)}$ \emph{spectrally negative} when $\theta > 0$ and \emph{spectrally positive} when $\theta < 0$.
\end{remark}

Our approach is to cast the problem for $R$ into that of $Y^{(x)}$ by using the relation shown below.  See also Figure \ref{fig_LLR_CUSUM} for graphical illustrations of the link between $R$ and $Y^{(x)}$.
\begin{figure}[htbp]
\begin{center}
\begin{minipage}{1.0\textwidth}
\centering
\begin{tabular}{cc}
 \includegraphics[scale=0.5]{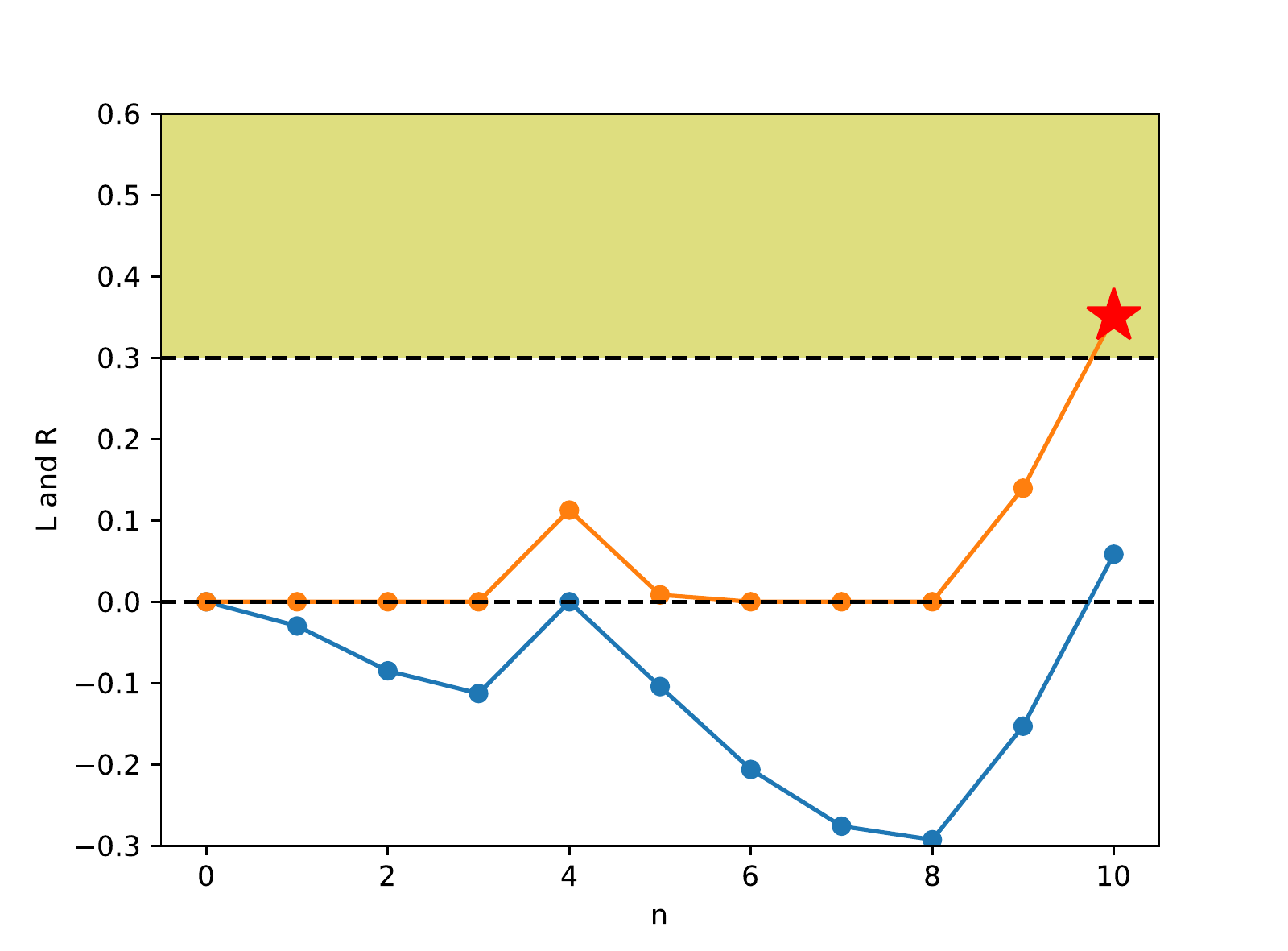} & \includegraphics[scale=0.5]{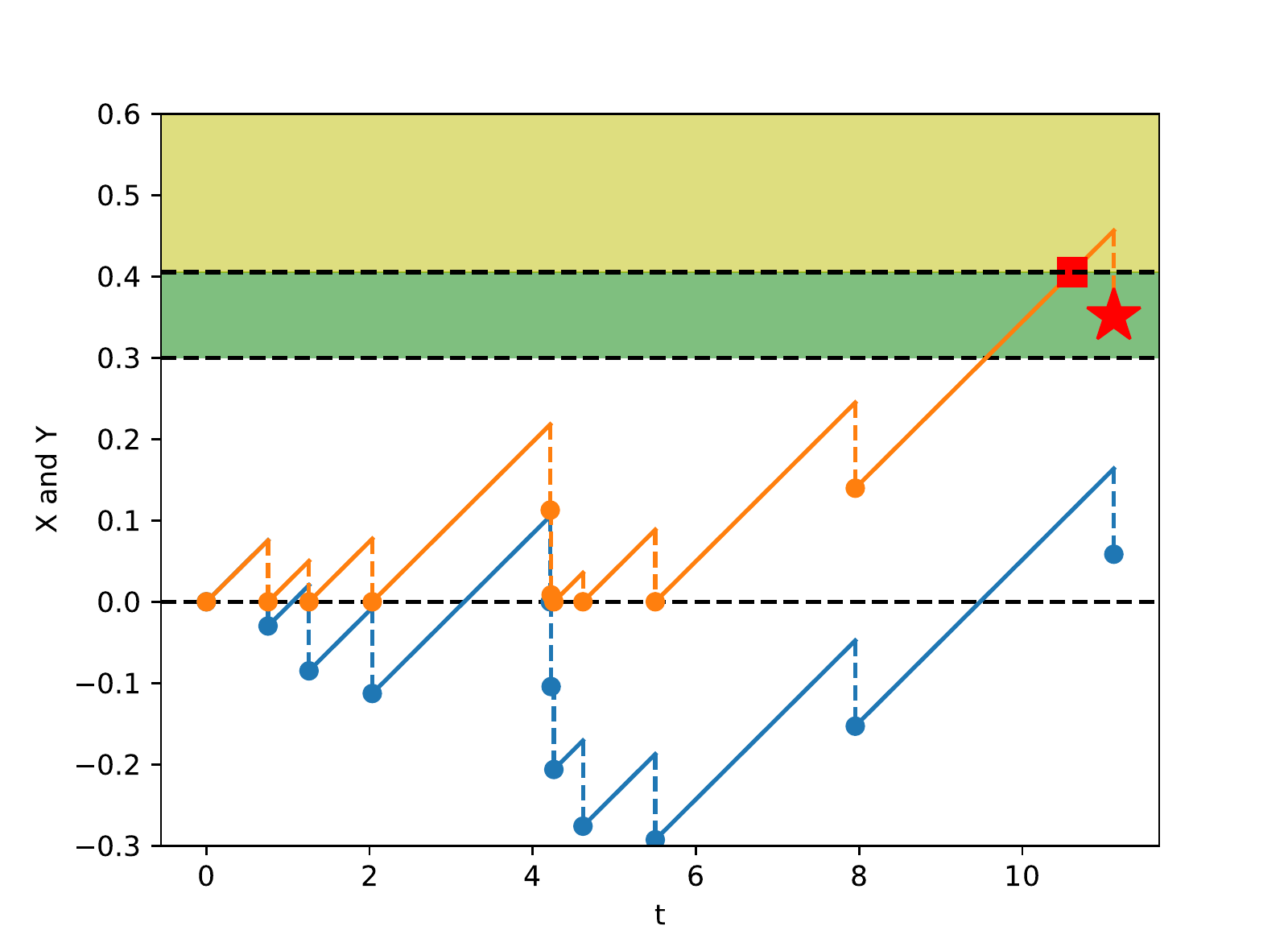}  \\
$n \mapsto L_n, R_n$ when $\theta > 0$ &  $t \mapsto X_t, Y_t$ when $\theta > 0$ \\
 \includegraphics[scale=0.5]{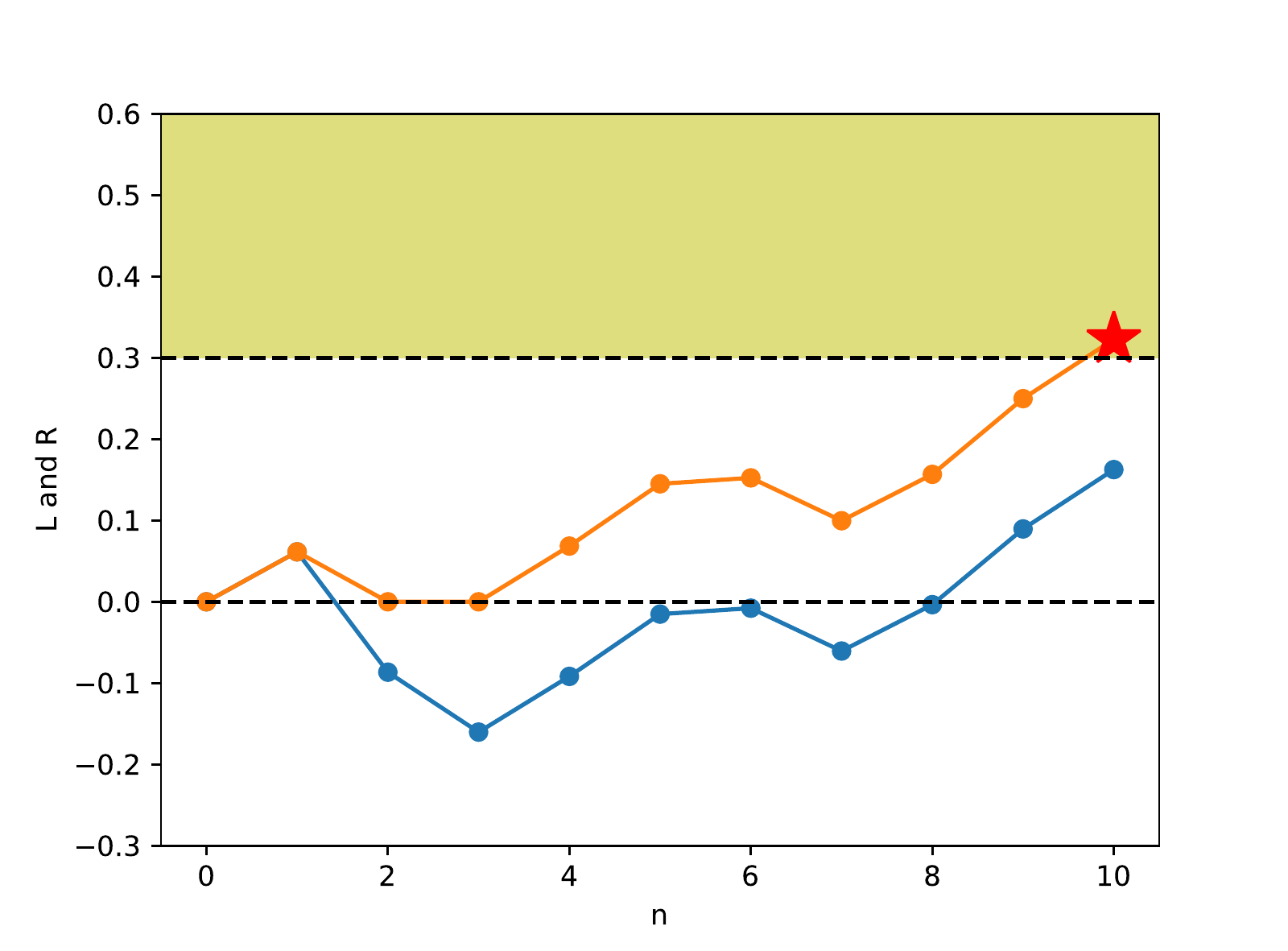} & \includegraphics[scale=0.5]{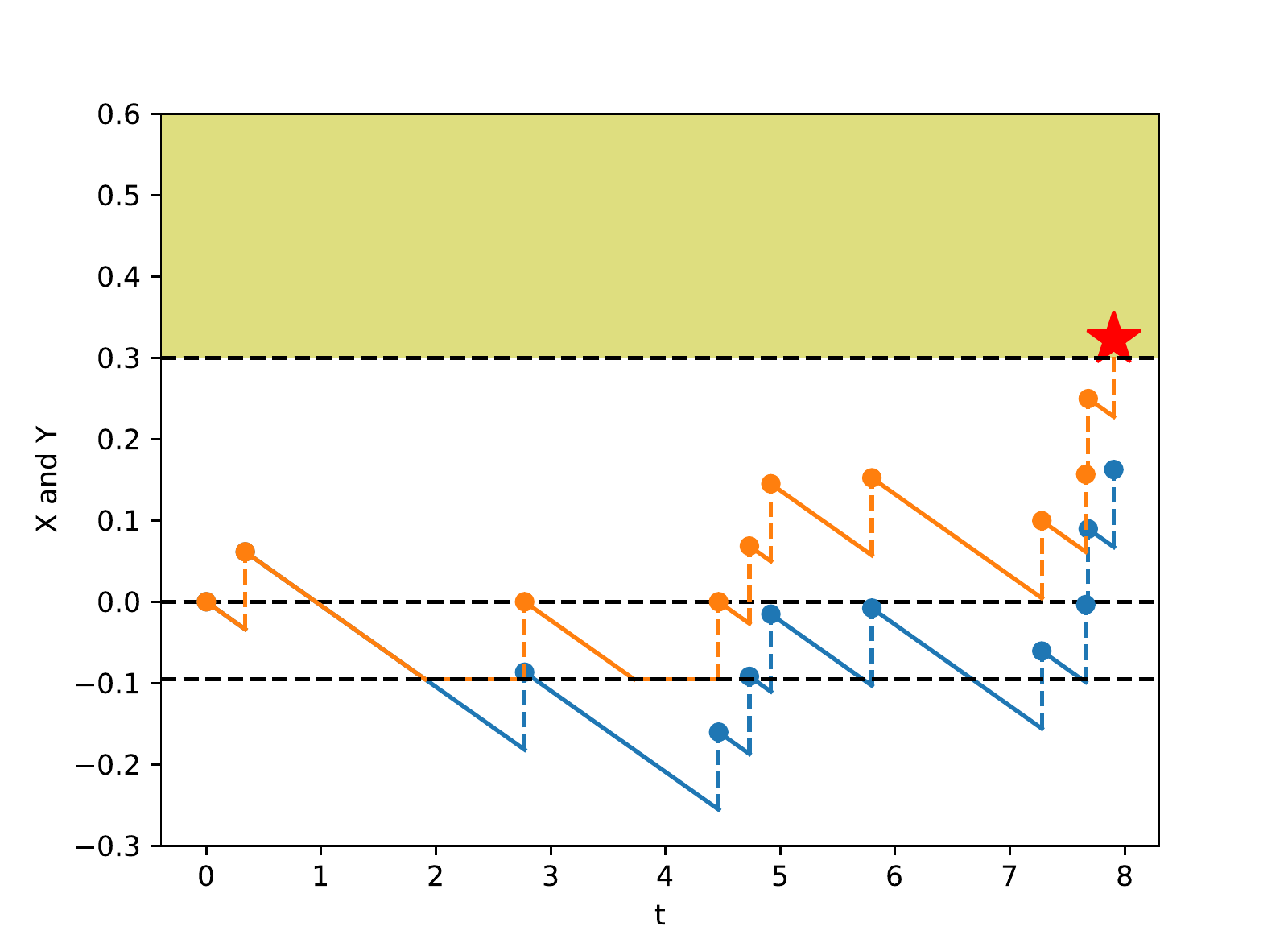}  \\
$n \mapsto L_n, R_n$ when $\theta < 0$ &    $t \mapsto X_t^{(|\kappa(\theta)|)} - |\kappa(\theta)|, Y_t^{(|\kappa(\theta)|)} - |\kappa(\theta)|$ when $\theta < 0$
 \end{tabular}
\end{minipage}
\caption{\small Sample paths of LLR (blue) and CUSUM (orange) and the corresponding generalized Sparre-Andersen process (blue) and its reflected process (orange) for $\theta  > 0$  and $\theta < 0$ when $A = 0.3$. (Top)
The LLR process $L$ and the CUSUM statistic $R$ on the left plot and the corresponding generalized Sparre-Andersen process $X$ and its reflected process $Y$ on the right plot for the case $\theta > 0$.   The star indicates the time when $R$ exceeds $A$ for the first time.  In the right plot, the square shows the first time $Y$ goes above $A + \kappa(\theta) = 0.40536$ and the star shows the first jump time afterwards. 
Note that $R_n = Y_{\eta_n}$ for all $n \geq 0$.
(Bottom) Similar plots for the case $\theta < 0$, except that
%The LLR process $S$ and the CUSUM statistic $R$ on the left plot and the corresponding generalized Sparre-Andersen process $X$ and its reflected process $Y$ on the right plot.  
the right figure shows $X^{(|\kappa(\theta)|)} - |\kappa(\theta)|, Y^{(|\kappa(\theta)|)} - |\kappa(\theta)|$, which are translations of $X^{(|\kappa(\theta)|)}$ and $Y^{(|\kappa(\theta)|)}$. The time $\tau^{(|\kappa(\theta)|)}_{A+ |\kappa(\theta)|}$ indicated by the star is the same as the time $Y^{(|\kappa(\theta)|)} - |\kappa(\theta)|$ crosses $A$.  Note that $R_n 
= Y_{\eta_n}^{(|\kappa(\theta)|)} - |\kappa(\theta)|$ for $n \geq 0$. 
%\jev{Kazu, if it is not difficult to redo the pictures, it would be nice to have sample paths crossing $A$ by a substantial amount. It looks like the square occurs at the same time as the star.}
}  \label{fig_LLR_CUSUM}
\end{center}
\end{figure}

\begin{proposition} \label{proposition_equivalence}
Fix $A > 0$ and let $\zeta$ be any strictly positive sequence.  The following holds a.s.
%\kazu{To be checked. It should hold for any general counting process with general interarrival times $\zeta$. do not have to be independent either.}
\begin{enumerate}
\item
When $\theta > 0$, $R_n = Y_{\eta_n}$ for $n \geq 0$ and
$T_A
= 1 + N_{\tau_{A+ \kappa(\theta)}}$.
\item When $\theta < 0$,
$R_n 
= Y_{\eta_n}^{(|\kappa(\theta)|)} - |\kappa(\theta)|$ for $n \geq 0$ and
$T_A = N_{\tau^{(|\kappa(\theta)|)}_{A+ |\kappa(\theta)|}}$.
\end{enumerate} 
\end{proposition}
\begin{proof}

%Define the reflected process of the Sparre-Andersen process $X^{(x)}$ (started at $x$) with a lower barrier $0$:
%TODO
%with its first passage time
%\[
%\tau^{(x)}_\beta := \inf \{ t \geq 0: Y_t^{(x)}  > \beta \}, \quad \beta > 0.
%\]
%We write $X= X^{(0)}$.
%Also define jump times of  the renewal process $N$ (or $X$) by
%\[
%\eta_l =\zeta_1 + \zeta_2 + \cdots + \zeta_l, \quad l \geq 1.
%\]

% and so as the laws of infimum
%$(\underline{S}_n)_{n \geq 1}$ and $(\min_{1 \leq k \leq n} X_{\xi_k} \wedge 0)_{n \geq 1}$.
%\kazu{below, we may need to assume the barrier $A$ is higher than the jump size $\kappa (\theta)$}

(1) Suppose $\theta > 0$  (and then $\kappa(\theta) > 0$). In view of Remark \ref{remark_sign}(1), the running infimum process
%of the Sparre-Andersen process $X$, which is denoted by 
$\underline X$ is updated only immediately after (negative) jumps and hence $\min_{1 \leq k \leq n} X_{\eta_k} \wedge 0 = \underline{X}_{\eta_n}$ implying, together with \eqref{connection_S_X},  $\underline{L}_n =  \underline{X}_{\eta_n}$ for all $n \geq 0$.
Therefore \eqref{R_n} becomes
\[
R_n =  X_{\eta_n} - \underline{X}_{\eta_n}  = Y_{\eta_n}, \quad n \geq 0.
\]
%Because the jump size $- \kappa(\theta) < 0$ is a deterministic constant, 
%and $\{ X_{\eta_n} > \underline{X}_{\eta_n} \} \subset \{ \underline{X}_{\eta_n-}  = \underline{X}_{\eta_n} \}$, %we can also write
%\begin{align*}
%R_n 
%&= \left\{ \begin{array}{ll} 0 & \textrm{ if } X_{\eta_n} = \underline{X}_{\eta_n} \\ X_{\eta_n-} - \kappa(\theta)  - \underline{X}_{\eta_n}  > 0 & \textrm{ if } X_{\eta_n} > \underline{X}_{\eta_n} \end{array} \right\} \\
% &= \left\{ \begin{array}{ll} 0 & \textrm{ if } X_{\eta_n} = \underline{X}_{\eta_n} \\ X_{\eta_n-} - \kappa(\theta)  - \underline{X}_{\eta_n-} > 0   & \textrm{ if }  \underline{X}_{\eta_n-}  = \underline{X}_{\eta_n} \end{array} \right\} 
%%=  X_{\eta_n-} - \underline{X}_{\eta_n}  - \kappa(\lambda) 
%= \max \Big(0, Y_{\eta_n-} - \kappa(\theta) \Big).
%\end{align*}
%In particular, 
If $R_n = Y_{\eta_n} > 0$ then no reflection is made at $\eta_n$ (jump size is exactly $-\kappa(\theta)$) and necessarily  $R_n = Y_{\eta_n}= Y_{\eta_n-} - \kappa(\theta) > 0$. 
Therefore, for $A > 0$, using that $X$ has only negative jumps,
\begin{align*}
%\inf \{n \geq 1: R_n > A \} 
T_A
&= \inf \{ n \geq 1: Y_{\eta_n-} - \kappa(\theta) > A \} 
\\ &= 1+ \# \textrm{\{jumps of $X$ (or $N$) before the instance $Y$ exceeds $A + \kappa(\theta)$\}} \\
&= 1 + N_{\tau_{A+ \kappa(\theta)}}
\end{align*}
where the addition of $1$ is needed because we also need to count the last observation, which is the jump occurring after $Y$ crosses $A + \kappa(\theta)$ upward (and then lands on somewhere above $A$); see Figure \ref{fig_LLR_CUSUM}. 

%On the other hand, because for $X$ its maximum is updated immediately before negative jumps, and hence 
%\[
%R_n = X_{\xi_n-} - \underline{X}_{\xi_n-} \wedge 0 - \kappa(\theta).
%\]
%Therefore the first passage time of $R$ above $x$ is the first time the reflected process of $X$ goes above $x + \kappa(\theta)$.

(2)  Suppose $\theta < 0$ (and then $\kappa(\theta) < 0$). 
%We have
%\[
%S_n =  X_{\eta_n-} - \kappa(\theta) =  X_{\eta_n-} + |\kappa(\theta)|.
%\]
Because $X_{\eta_n} = X_{\eta_n-} - \kappa(\theta) =  X_{\eta_n-} + |\kappa(\theta)| = X_{\eta_n-}^{|\kappa(\theta)|}$, by \eqref{connection_S_X},
\[
\underline L_n = \min_{1 \leq k \leq n} X_{\eta_k} \wedge 0 %=  \min_{1 \leq k \leq n} (X_{\eta_k-} + |\kappa(\theta)|) \wedge 0
 =  \min_{1 \leq k \leq n} (X_{\eta_k-}^{(|\kappa(\theta)|)}) \wedge 0.
% = |\kappa(\theta)|+ \min_{1 \leq k \leq n} X_{\eta_k-} \wedge (- |\kappa(\theta)|).
\]
 In view of Remark \ref{remark_sign}(2), $X$ has a negative drift with positive jumps and hence
 %the running infimum of $X$ is updated only immediately before (positive) jumps and hence
$\min_{1 \leq k \leq n} X_{\eta_k-}^{(|\kappa(\theta)|)} =  \underline{X}_{\eta_n}^{(|\kappa(\theta)|)}$ and therefore
$\underline L_n = 
% |\kappa(\theta)|+ \underline{X}_{\eta_n} \wedge (- |\kappa(\theta)|) = 
 \underline{X}_{\eta_n}^{(|\kappa(\theta)|)} \wedge 0$. Substituting this in \eqref{R_n},
\[
R_n  = X_{\eta_n} - 
\underline{X}_{\eta_n}^{(|\kappa(\theta)|)} \wedge 0  = X_{\eta_n}^{(|\kappa(\theta)|)} - |\kappa(\theta)| -
\underline{X}_{\eta_n}^{(|\kappa(\theta)|)} \wedge 0
= Y_{\eta_n}^{(|\kappa(\theta)|)} - |\kappa(\theta)|, \quad n \geq 0.
% \big[ |\kappa(\theta)|+ \underline{X}_{\eta_n} \wedge (- |\kappa(\theta)|) \big].
% = Y_{\eta_n}^{(- |\kappa(\theta)|)} - |\kappa(\theta)|.
\]
Therefore (again see Figure \ref{fig_LLR_CUSUM}),
\begin{align*}
%\inf \{n \geq 1: R_n > A \}
T_A
%&= \inf \{ n \geq 1:  X_{\eta_n} -  \big[ |\kappa(\theta)|+ \underline{X}_{\eta_n} \wedge (- |\kappa(\theta)|) \big] > A \} 
%\\
%&= \inf \{ n \geq 1:  X_{\eta_n} -  \big[ (|\kappa(\theta)|+ \underline{X}_{\eta_n} ) \wedge 0 \big] > A \} 
%\\
%&= \inf \{ n \geq 1:  |\kappa(\theta)| + X_{\eta_n} -  \big[ (|\kappa(\theta)|+ \underline{X}_{\eta_n} ) \wedge 0 \big] > A + |\kappa(\theta)| \} 
%\\
 &= \inf \{ n \geq 1:  Y_{\eta_n}^{(|\kappa(\theta)|)} - |\kappa(\theta)|  > A \} 
\\ &= \# \textrm{\{jumps of $X$ (or $N$) before or at the instance $Y^{(|\kappa(\theta)|)}$ exceeds $A+ |\kappa(\theta)|$\}} \\
&= N_{\tau^{(|\kappa(\theta)|)}_{A+ |\kappa(\theta)|}}.
\end{align*}

%
%and hence
%\[
%(\underline{S}_n)_{n \geq 1} \sim  (\underline{X}_{\_n-} + \kappa(\theta) ) \wedge 0.
%\]
%and hence
%\begin{align*}
%(\underline{R}_n)_{n \geq 1} &\sim  X_{\xi_n} - (\underline{X}_{\xi_n-} + \kappa(\theta) ) \wedge 0 \\
%%&= X_{\xi_n-} + \kappa (\theta) - (\underline{X}_{\xi_n-} + \kappa(\theta) ) \wedge 0  \\
%%&= X_{\xi_n-} - (\underline{X}_{\xi_n-}  \wedge (-\kappa(\theta)) )
% &\sim  X_{\xi_n} - (\underline{X}_{\xi_n-} ) \wedge ( - \kappa(\theta)) - \kappa(\theta)
%\end{align*}
%On the other hand, because for $X$ its maximum is updated immediately after positive jumps, and hence 
% the first passage time of $R$ above $x$ is the first time the reflected process of $X$ ($X - \underline{X}^{(-\kappa(\theta))}$) goes above $x + \kappa(\theta)$.
\end{proof}

\section{First passage analysis of Sparre-Andersen processes with phase-type interarrivals} \label{sec_fluctuations_SA}

In the last section, we discussed in Proposition \ref{proposition_equivalence} that the CUSUM statistic can be written in terms of  the reflection of  the process \eqref{Sparre-andersen_CUSUM}, whose interarrival times are given by the observation $\zeta$. In particular, this reduces to an ordinary Sparre-Andersen process if the observations $\zeta$ are i.i.d.
In this section, we derive new identities in the fluctuation theory of Sparre-Andersen processes with PH interarrivals. Although our main motivation of this section is its application in the computation of the optimal barrier \eqref{lorden_beta} in the minimax formulation (see Section \ref{subsection_SA_CUSUM}), we consider a wider class of Sparre-Andersen processes, not necessarily with jumps of constant size, which have applications in research areas beyond the study of the CUSUM procedure.
These results are further generalized to non-i.i.d.\ settings in Section \ref{section_non_iid}.
%We generalize the results to the case with change point in the next section and for the analysis of the CUSUM strategy described in the previous section.

%\red{[change $n$ to $\mathfrak{n}$]}
We denote, by $\mathcal{PH}(E, \bs \alpha, \bs T, \bs t)$, a PH distribution with representation $(E, \ba,\bs T, \bs t)$.  In other words, it is the distribution of the first absorption time of a finite-state continuous-time Markov chain on the state space $E \cup \{ \Delta \}$, consisting of the set of $\mathfrak{n} (\geq 1)$ transient states $E :=  (1,2,\ldots, \mathfrak{n})$ and a single absorbing state $\Delta$. Its initial distribution on $E$ is given by the $\mathfrak{n}$-dimensional row vector $\bs \alpha = (\bs \alpha_1, \ldots, \bs \alpha_{\mathfrak{n}})$ (the probability of starting at $\Delta$ is zero) such that $\sum_{i=1}^{\mathfrak{n}} \ba_i = 1$ and  transition rate matrix is given by
\[
\left( \begin{array}{ll} \bs T & \bs t \\  \bs 0^\top & 0 \end{array} \right).
% \quad \textrm{ where }  \bs t = - \bs T \bs 1.
\]
Here, the sub-intensity matrix $\bs T$ shows the transition rates among those in $E$ and the exit rate vector $\bs t$ shows the rate of absorbing to $\Delta$ from each state in $E$.
We allow the Markov chain to be \emph{defective} in the sense that 
\begin{align}
\bs q:=-\bs T\bs 1-\bs t\geq \bs 0 \label{def_q}
\end{align}
 is not necessarily $\bs 0$; in other words, it is killed and sent to a cemetery state with rate $\bs q_i$ while it is in phase $i \in E$.
Here and throughout the paper, let $\bs 1 = (1,\ldots, 1)^\top$ and $\bs 0 = (0, \ldots, 0)^\top$ be the column vectors consisting of all ones and all zeros, respectively (with dimensions clear from the context). 

On a probability space $(\Omega, \mathcal{E}, \p)$, define the  Sparre-Andersen process 
\begin{equation}\label{eq:MAP}
X_t:= X_0 + \gamma t-S_{N_t}, \qquad t \geq 0, \end{equation}
where
\begin{align} \label{def_S}
 S_n:=\sum_{i=1}^n C_i, \quad n \geq 0. \
\end{align}
Here, we assume the drift is strictly positive ($\gamma > 0$),  $N = (N_t)_{t \geq 0}$ is 
a renewal process with independent $\mathcal{PH}(E, \bs \alpha, \bs T, \bs t)$-distributed interarrival times  (a.k.a.\ PH renewal process)  and  $C = (C_i)_{i \geq 1}$ is an i.i.d.\ sequence of $(0,\infty)$-valued random variables independent of $N$.

\begin{remark} \label{remark_SA_CUSUM}
The process \eqref{Sparre-andersen_CUSUM} for the analysis of the CUSUM statistic when $\theta > 0$ and $\zeta$ are i.i.d.\ is a special case of \eqref{eq:MAP} with $\gamma = \theta$, deterministic jumps of size $C \equiv \kappa(\theta) > 0$, and the observation $\zeta_n \sim \mathcal{PH}(E, \bs \alpha, \bs T, \bs t)$ for all $n \geq 1$. The case $\theta < 0$ can be dealt by considering its dual process (see Section \ref{subsection_SP}).
\end{remark}
It is a common practice to write \eqref{eq:MAP} as a (spectrally negative) MAP.
As in \cite[Example 1.1, Chap. XI]{APQ}, the renewal process $N$ can be 
described as the number of arrivals of a background Markov chain $J$ with transition rate matrix $\bs T+ \bs B$ where $\bs T$ and $\bs B := \bt\ba$  are the intensities of transitions without arrivals and with arrivals, respectively.  At each arrival that occurs with rate $\bt$, $N$ jumps up by one and $J$ is reset according to the  distribution $\ba$.
We refer the reader to \cite[Ch.\ XI]{APQ} for a review of Markov arrival processes.
% using the Markov chain with
%$(\bs \alpha, \bs T)$ modified such that it restarts as soon as it is absorbed to $\Delta$. 
%Let $(N,J) = (N_t,J_t)_{t \geq 0}$ be  the Markovian arrival process, where $J$ is a Markov chain (with killing) with transition rate matrix $\bs T+\bt\ba$ 
%and the intensities $\bt\ba$ correspond to arrivals in the language of Markovian arrival processes~\cite[Ch.\ XI]{APQ}.
%In other words, it is a modification of the Markov chain describing $\mathcal{PH}(E, \bs \alpha, \bs T, \bs t)$ such that it restarts as soon as it is absorbed to $\Delta$. 
%Finally, consider a Sparre-Andersen process 
%\begin{equation}\label{eq:MAP}X_t:=\gamma t-S_{N_t},\qquad S_n:=\sum_{i=1}^n C_i,\; \gamma>0,\end{equation} where $C_i>0$ is an i.i.d.\ sequence (of claim sizes) independent of the renewal process~$N_t$.
With the background process $J$ as a modulator, we describe \eqref{eq:MAP} as the ordinator of the MAP  $(X,J)$, which experiences negative jumps of size $C$ upon arrivals (jump times of $N$). 

We let $\p_{x,i} (\cdot)$ (with parentheses) be the law of $(X,J)$ when $(X_0, J_0) = (x,i)$  for $x \in \mathbb{R}$ and $i \in E$. We also write $\p_x [A, J_\tau]$ (with brackets) for the $\mathfrak{n}\times \mathfrak{n}$ matrix whose $(i,j)$-th element is $\p_{x,i} (A, J_\tau = j)$, for any event $A$ and (random or deterministic) time $\tau$. In particular,  $\p_x [J_\tau]_{ij} = \p_{x,i} (J_\tau = j)$, $i,j \in E$.  Analogously,  we let $\mathbb{E}_x [Y; J_\tau]$ be the matrix of expectations of $\mathbb{E}_{x,i} (Y; J_\tau = j )$.  
We drop the subscript when $X_0 = 0$.  Different from \eqref{Sparre-andersen_CUSUM}, we omit the superscript $(x)$ for the starting value, which can be modeled by using the measure $\mathbb{P}_x$.

\begin{remark}
Note that \eqref{eq:MAP} is more general than \eqref{Sparre-andersen_CUSUM}, and to avoid confusion 
we use different fonts for the probability/expectation operators from those in Section \ref{subsection_change_point}. 
\end{remark}

A so-called matrix exponent of the MAP $(X,J)$ is then given by
\begin{align}
\bs F(s):=\gamma s\bs I_{\mathfrak{n}}+\bs T+\e (e^{-s C_1}) \bs B, \quad s \geq 0, \label{def_F}
\end{align}
and it satisfies
\[
\e [e^{s (X_t-X_0)};J_t ]  =e^{\bs F(s)t}, \quad t \geq 0,
\]
where $\bs I_{\mathfrak{n}}$ is the $\mathfrak{n} \times \mathfrak{n}$ identity matrix.
By convention it is assumed that $X$ is killed when $J$ is killed (sent to a cemetery state), which occurs with rate vector $\bs q$ as in \eqref{def_q}.

%We let $\mathbb{P}_x$ be the law of $X$ when $X_0 = x$ and $\tilde{\mathbb{P}}_x$ be the law of $\tilde{X}$ when $\tilde{X}_0 = x$. We omit the subscript when $x = 0$.

%\subsection{The underlying Markov additive process}

\subsection{Fluctuation theory of Sparre-Andersen process}

%\subsection{Matrix-valued scale function}
There is a rich fluctuation theory for spectrally negative MAPs~\cite{breuer2008first,ccauglar2021optimal, IP12,kyprianou2008fluctuations}, and the basic object underlying various identities is a so-called scale matrix $\bs W: [0,\infty) \to \R^{\mathfrak{n} \times \mathfrak{n}}$.
This continuous, right-differentiable, matrix-valued function is characterized by the transform:
\[\int_0^\infty e^{-s x}\bs W(x)\D x=\bs F(s)^{-1}\]
for $s > \max\{\Re(z): z \in \mathbb{C}, \det (\bs F(z)) = 0 \}$. See, e.g., \cite[Thm.\ 1]{IP12}. Moreover, $\bs W(x)$ is invertible for $x>0$.  We also write its integral $\overline {\bs W}(x):=\int_0^x \bs W(z)\D z$ and right-hand derivative $\bs W'_+(x)$ for $x \geq 0$.

% always exists for all $x > 0$ (see Theorem 5 of \cite{IP12}).
\begin{remark}
In the following, we use several results of ~\cite{IP12} where the Markov chain $J$ is assumed to be irreducible, which is not the case below in this paper. 
It can be checked that this assumption is indeed redundant, given that the quantity
\[\psi(s)=\max_i\{\Re(\lambda_i(s)): \lambda_i(s)\text{ is an eigenvalue of } \bs F(s)\}
\] is treated with some care. 
%\kazu{changed the notation as $\kappa$ was already used} 
In general, it should not be called the Perron-Frobenius eigenvalue, and we should not rely on $\psi'(0)$ or the asymptotic drift concept. 
In particular, all the results in~\cite{IP12} apart from Cor.\ 4 (in the given form) hold without irreducibility assumption.
\end{remark}

We refer the reader to \cite{IP12} for a list of expectations one can compute using the scale matrix. Here, we focus on the identities relevant to the  performance measures of the CUSUM procedure.

%\kazu{TODO: rename $A$, $x$ etc}

As in \eqref{reflected_process} and \eqref{tau_B}, we define the reflected process
\begin{align*}
Y_t := X_t - \underline{X}_t \wedge 0, \quad t \geq 0, 
\end{align*}
where $\underline{X}_t := \inf_{0 \leq s \leq t} X_s$, and its first passage time
\begin{align}
\tau_a :=\inf\{t\geq 0: Y_t> a\},\qquad  a > 0.  \label{tau_a_def}
\end{align}

%Define the reflected process \red{$Y^{(x)}$} of \eqref{eq:MAP} and its first passage time \red{$\tau_a^{(x)}$} above as in  
For $a > 0$, by \cite[Thm.\ 2]{IP12},
\begin{align}
\p [J_{\tau_a} ]=  \Big(\bs I_{\mathfrak{n}}-\overline {\bs W}(a) \bs{F}(0)\Big)^{-1} = \Big(\bs I_{\mathfrak{n}}-\overline {\bs W}(a)(\bs{T}+\bs B)\Big)^{-1}. \label{J_at_passage}
\end{align} 
%\kazu{change from $x$ to $a$?}
Let the number of arrivals coming from phase $k \in E$ counted until $\tau_a$ be denoted by
\begin{align}
N_{\tau_a}(k) := \sum_{t \leq \tau_a: \Delta N_t \neq 0} 1_{\{ J_{t-} = k \}}. \label{N_sub}
\end{align}
The following results can be derived easily by writing its generating function in terms of the scale matrix via \eqref{J_at_passage}. Its proof is deferred to Appendix \ref{proof_lemma_expectation}.
% the proof is given in the appendix. 
%\red{We also let  $\mathbb{E}^{\ba} (\cdot) := \mathbb{E} (\cdot | X_0 = 0, J_0 \sim \ba) = \ba \mathbb{E} [\cdot; J_{\tau_a}] \bs 1$.}

%\jev{checked. In the defective case the expressions will be slightly more complicated.}
\begin{lemma} \label{lemma_expectation}
Suppose $J$ is non-defective (i.e. $\bs q = \bs 0$).
For $a > 0$, we have
%\[\e_{\ba} (N_{\tau_x}(k)) = \Big[\ba \Big(I- \overline{\bs W}(x) (T+\bt\ba)\Big)^{-1}  \overline{\bs W}(x)\Big]_{k}  \bt_{k}.
%\]
\[\e_{0,i} (N_{\tau_a}(k)) = \Big[ \Big(\bs I_{\mathfrak{n}} - \overline{\bs W}(a) (\bs T+ \bs B)\Big)^{-1}  \overline{\bs W}(a)  \mathrm{diag} (\bt) \Big]_{ik}, \quad i, k \in E,
\]
%%\[\e [N_{\tau_x
%%\[\e [N_{\tau_x} ; J_{\tau_x} ]= \Big(\bs I_n-\overline {\bs W}(x)(\bs T+\bt\ba)\Big)^{-1}\overline {\bs W}(x) \mathrm{diag}(\bt),\]
%\red{
%\[(\e_{\ba} (N_{\tau_x}(i,j) ) )_{ij} = \ba \Big(\bs I_n-\overline {\bs W}(x)(\bs T+\bt\ba)\Big)^{-1}\overline {\bs W}(x) \circ (\bt \ba),\]
%%where $N_{\tau_x}(j)$ which is the number of arrivals occurring in phase $j$ until $\tau_x$ [check]
%}
%where $\circ$ is the Hadamard product.
% and
% and
%$\mathrm{diag}(\bt)$ is the matrix with diagonal elements consisted of $\bt$.
and hence the unconditional expected number of arrivals until $\tau_a$ is
\[
%\e^{\ba} (N_{\tau_a} ) 
\mathbb{E} \big(N_{\tau_a}  | J_0 \sim \ba \big) = \sum_{i,k \in E} \ba_i \e_{0,i} (N_{\tau_a}(k))  =\ba \Big(\bs I_{\mathfrak{n}}-\overline {\bs W}(a)(\bs T+ \bs B )\Big)^{-1}\overline {\bs W}(a)\bt.\]
%\red{[to be changed to avoid the confusion with matrix notation]}
\end{lemma}

\subsection{Spectrally positive case} \label{subsection_SP}
By flipping the process \eqref{eq:MAP}, we can also consider the case with a negative drift and positive jumps.  Suppose temporarily that, with the same $S$ and $N$ as in \eqref{eq:MAP},
\begin{align}
\textrm{(SP)} \qquad X_t = X_0 - \gamma t + S_{N_t}, \quad t \geq 0.  \label{SP}
\end{align} Then its dual process $X_t^d :=-X_t$  started at $X_0^d = -X_0$ admits the form \eqref{eq:MAP}. Let $\bs W^d(x)$ be the scale matrix for the spectrally negative MAP $(X^d, J)$. Below it is understood that  $\tau_a$ and \eqref{N_sub} are for the original spectrally positive MAP $(X,J)$ and $\mathbb{P}_{x,i}$ is its law when $X_0 = x$ (starting point of the original spectrally positive MAP) and $J_0 = i$. Different from the spectrally negative case above, here we compute the first passage identities for a general starting point because for the analysis of CUSUM, we need the case the starting point is different from the reflection barrier $0$ (see Proposition \ref{proposition_equivalence}(2)).
%We define the reflected process and its hitting time analogously to  \eqref{reflected_process} and \eqref{tau_B} and let the latter denoted by $\tau_B^d$.

According to~\cite[Thm.\ 6]{IP12}, because $\bs T$ and $\bs B$ remain the same for $X^d$ and $X$,
% the matrix exponent of $X^d$  is of the same form as \eqref{def_F},
% because $(\bs T+\bt\ba) \bs 1 = \bs 0$
\begin{align} \label{first_passage_SP}
\p_{x} [J_{\tau_a}] =  \bs I_{\mathfrak{n}} - \Big( \overline{\bs W}^d(a-x)  - \bs W^d(a-x)
(\bs W^{d})'_+(a)^{-1}
{\bs W}^d(a) \Big) 
\big(\bs T+ \bs B \big), \quad a > 0, \; 0 \leq x \leq a.
%\bs{F}^d(0).
\end{align}
The following is a direct consequence of this identity and hence we defer its proof to Appendix \ref{proof_lemma_expectation_dual}. 

\begin{lemma} \label{lemma_expectation_dual} 
Suppose (SP) as in \eqref{SP} and $J$ is non-defective (i.e. $\bs q = \bs 0$).
For $a > 0$ and $0 \leq x \leq a$,
%we have
%\[\e_y [N_{\tilde{\tau}_x}; J_{\tau_x} ] =- \Big(\int_0^{x-y} \tilde {\bs W}(z)\D z-\tilde {\bs W}(x-y)\tilde {\bs W}'_+(x)^{-1}\tilde {\bs W}(x)\Big) \mathrm{diag} (\bt).\]
\[\e_{x,i} (N_{\tau_a}(k)) =-  \Big[\Big( \overline{\bs W}^d(a-x)- \bs W^d(a-x) (\bs W^{d})'_+(a)^{-1} {\bs W}^d(a)\Big) \mathrm{diag}(\bt) \Big]_{ik}, \quad i, k \in E, \]
%where $N_{\tau_x}(j)$ which is the number of arrivals occurring in phase $j$ until $\tau_x$ [check]
and hence the unconditional expected number of arrivals until $\tau_a$ is
%Hence, %the unconditional expectation when $J_0 \sim \ba$ is
\[%\e_y^{\ba} (N_{\tau_a}  ) 
\mathbb{E}_x (N_{\tau_a}  | J_0 \sim \ba)=-\ba \Big( \overline{\bs W}^d(a-x)- \bs W^d(a-x) (\bs W^{d})'_+(a)^{-1} {\bs W}^d(a)\Big) \bt.\]
%\red{[to be changed to avoid the confusion with matrix notation]}
\end{lemma}

\subsection{The case of CUSUM} \label{subsection_SA_CUSUM}

Now recall our discussions in Section \ref{subsection_change_point}. The optimal barrier in the minimax formulation is given by $A_{\beta}$ such that \eqref{lorden_beta} holds, and for this computation we need the average run length $\mathrm{ARL}(T_A) \equiv \mathrm{E}_\infty [T_A]$ for $A > 0$. Here, we consider the case $\zeta_n \sim F_0 \sim \mathcal{PH}(E, \bs \alpha, \bs T, \bs t)$ for all $n \geq 1$, and this defines the Markov chain $J$. Recall again that any positive distribution can be approximated by PH distributions.
% and hence this assumption is not restrictive. 
%
%
% where under the law $\mathrm{P}_\infty$ corresponds to the case discussed in
%Remark \ref{remark_SA_CUSUM} with $F_0$ given by  
%the observation is independent and $F_0$-distributed.
\begin{enumerate}
\item
 Suppose $\theta > 0$. Let ${\bs W}$ be the scale matrix of the MAP $(X, J)$ as in \eqref{eq:MAP} with  $\gamma = \theta > 0$ and  $C \equiv \kappa(\theta) > 0$ (see Remark \ref{remark_SA_CUSUM}).
% Let us consider the scale matrix of the MAP as in Remark \ref{remark_SA_CUSUM} with $N$ being the renewal process with i.i.d.\ $F_0$-distributed interarrival times.  
 In view of Proposition \ref{proposition_equivalence}(1),  Lemma \ref{lemma_expectation} gives
\begin{align} \label{ARL_SN}
\begin{split}
\mathrm{ARL}(T_A) &= 1+ \e (N_{\tau_{A+\kappa(\theta)}} | J_0 \sim \ba) \\ &= 1+ \ba \Big(\bs I_{\mathfrak{n}} -\overline {\bs W}(A+\kappa(\theta))(\bs T+\bs B)\Big)^{-1}\overline {\bs W}(A+\kappa(\theta))\bt, \quad A > 0.
\end{split}
\end{align}
\item Suppose $\theta < 0$. Let ${\bs W}^d$ be the scale matrix
of the MAP $(X^d, J)$  as in \eqref{eq:MAP} with $\gamma = -\theta > 0$ and  $C \equiv |\kappa(\theta)| =  -\kappa(\theta) > 0$ (again see Remark \ref{remark_SA_CUSUM}).
%Then we consider \eqref{SP} with $\gamma = -\theta$ and the scale matrix of its dual given by $\bs W^d$.
In view of Proposition \ref{proposition_equivalence}(2),  Lemma \ref{lemma_expectation_dual}  gives
\begin{align}  \label{ARL_SP}
\begin{split}
\mathrm{ARL}(T_A)  &= \e_{|\kappa(\theta)|} (N_{\tau_{A+|\kappa(\theta)|}}| J_0 \sim \ba) \\ &=-\ba \Big( \overline{\bs W}^d(A)- \bs W^d(A) (\bs W^{d})'_+(A+ |\kappa(\theta)|)^{-1} {\bs W}^d(A+ |\kappa(\theta)|)\Big) \bt, \quad A > 0.
%\\ -\ba \Big(\int_0^{A} \tilde {\bs W}(z)\D z -\tilde {\bs W}(A)\tilde {\bs W}'_+(A+ |\kappa(\lambda)|)^{-1}\tilde {\bs W}(A+ |\kappa(\lambda)|)\Big) \bt.
\end{split}
\end{align}
\end{enumerate}

\section{Series expansion of the scale matrix} \label{section_series_expansion}
%\red{[changed this from subsection to section]}
As discussed in the previous section, the computation of the identities of interest boils down to that of the scale matrix. Here, we derive a new formula for $\bs W$ of the spectrally negative MAP $(X,J)$ of the form \eqref{eq:MAP}, generalizing the series expansion in~\cite[Thm.\ 2.2]{LW20} for the Cram\'er-Lundberg model (i.e.\ $N$ is a Poisson process) and also~\cite[Thm.\ 2]{AAI14} in the case of Erlang interarrival times and deterministic jumps.
These previous results were obtained by transform inversion, which becomes infeasible in this more general setting. Hence, we take a different approach.

For every $k\geq 1$ we define an $\mathfrak{n} k\times \mathfrak{n}k$ sub-transition rate matrix $\bs T_k$ and $\mathfrak{n} k \times \mathfrak{n}$ matrices $\bs e_{k1}$ and  $\bs e_{kk}$:
\begin{align} \label{def_T_k}
\bs T_k=\begin{pmatrix}
\bs T& \bs B &\ldots&\bs O \\
 &\ddots&\ddots&\vdots\\
\vdots&  &\bs T &\bs B\\
\bs O&\ldots& &\bs T
\end{pmatrix},
\qquad \bs e_{k1}=\begin{pmatrix}
\bs I_{\mathfrak{n}}\\ \bs O \\ \vdots\\ \bs O
\end{pmatrix},
\qquad \bs e_{kk}=\begin{pmatrix}
\bs O \\ \vdots\\\bs O\\ \bs I_{\mathfrak{n}}
\end{pmatrix},
\end{align}
so that $\bs T_1=\bs T$ and $ \bs e_{11} = \bs I_{\mathfrak{n}}$. Here and for the rest of the paper, $\bs O$ is a zero matrix of appropriate dimension. These definitions are motivated by the identity
\begin{equation}\label{eq:Tk}
\p[N_t=k-1,J_t ] =\bs e^\top_{k1}e^{\bs T_k t}\bs e_{kk},\qquad t\geq 0, \; k \geq 1,
\end{equation}
which gives the matrix of probabilities of seeing $k-1$ arrivals by time $t$ and being in a particular phase at this time.  %Here and for the rest of the paper, it is understood that $\p_x (A, J_t)$

\begin{theorem}\label{thm:W}
The scale matrix of a MAP in~\eqref{eq:MAP} has the representation
\[\bs W(x)=\frac{1}{\gamma}\sum_{k\geq 1}\int_{y\in[0,x]} \bs e^\top_{k1}e^{\bs T_k (y-x)/\gamma}\bs e_{kk} \p(S_{k-1}\in \D y), \qquad x\geq 0,\]
where $S$ is defined as in \eqref{def_S} with the understanding that $S_0 = 0$. 
\end{theorem}
It is important to point out that the above series is absolutely convergent for any $x\geq 0$.
Indeed, consider the matrix norm $\|\bs M\|:=\max_i\sum_j |\bs M_{ij}|$, and note that $\|\bs T_k\|\leq\|\bs T\|+\| \bs B \|=:r$, which is independent of~$k$.
Thus $\|\bs e^\top_{k1}e^{\bs T_k (y-x)/\gamma}\bs e_{kk}\|\leq e^{r|y-x|/\gamma}$ and also
\begin{equation}\label{eq:normbound}\|\bs W(x)\|\leq \frac{1}{\gamma}\sum_{k\geq 1}\int_{y\in[0,x]} e^{r|y-x|/\gamma}\p(S_{k-1}\in \D y)
\leq \frac{1}{\gamma}e^{rx/\gamma}\sum_{k\geq 1}\p(S_{k-1}\leq x)<\infty,\end{equation}
where finiteness of the latter sum follows from the basic renewal theory.

\begin{remark}
The term $\bs e^\top_{k1}e^{\bs T_k (y-x)/\gamma}\bs e_{kk}$, $k \geq 1$, is the top-right corner block of $e^{\bs T_k (y-x)/\gamma}$, and can be written in an alternative way to avoid high dimensionality when computing it; see Appendix \ref{remark_matrix_exp}.
\end{remark}

\begin{proof}[Proof of Theorem~\ref{thm:W}]

First, we assume that $(X,J)$ is a killed process with $\bs q_i = q > 0$ for all $i \in E$.
%(excluding the oscillating case would be enough). \kazu{you mean we assume $\bs q>0$ or just assume that $X$ does not oscillate?}
As in ~\cite[Thm.\ 1 and (10)]{IP12} (see also \cite[Eq.\ (12)]{IP12}), we can write
\begin{equation}\label{eq:Wrep}\bs W(x)=e^{-\bs G x}\bs H(0)-\bs H(-x),\qquad x>0,\end{equation}
where $\bs G$ is the transition rate matrix of the first passage Markov chain (i.e.\ $\p[J(\sigma_x)] = \exp (\bs G x)$ for $x > 0$ where $\sigma_x := \inf \{ t > 0: X_t > x \}$)
 and $\bs H(x)$ denotes the matrix of expected occupation times at the level~$x$; see \cite[Sec.\ 4]{IP12} for the precise definitions. 
%We use the identity \eqref{eq:Wrep} for the computation of the scale matrix. Here, recall $\bs G$ is the transition rate matrix of the first passage Markov chain 
% and $\bs H(x)$ denotes the matrix of expected occupation times at the level~$x$.  
 In the present setting (where $X$ is of bounded variation and $0$ is irregular for itself), $\gamma\bs H(x)_{ij}$ is the expected number of times the level $x$ is hit in phase $j$ when starting in phase~$i$. Moreover, there is a standard identity 
\begin{equation}\label{eq:GH}\bs H(x) = e^{\bs G x}\bs H(0) ,\qquad x\geq 0,
\end{equation} which follows by the strong Markov property and additivity of occupation times together with the lack of positive jumps of $X$.

Consider the $\mathfrak{n} \times \mathfrak{n}$ matrix of probabilities of hitting the level $x\in \R$ in stage $k\geq 1$ (between $k-1$ and $k$th arrivals): 
\[
\int_{y\geq 0, y+x\geq 0} \bs e^\top_{k1}e^{\bs T_k (y+x)/\gamma}\bs e_{kk} \p(S_{k-1}\in \D y),
\]
where the $i$-th row corresponds to starting in phase $i$ and the $j$-th column to hitting in phase~$j$. 
This identity readily follows by conditioning on $S_{k-1}\geq 0$ which is independent of the rest, and then applying~\eqref{eq:Tk}. 
By summing up this expression over $k$, we get 
\[\bs H(x)=\frac{1}{\gamma}\sum_{k\geq 1}\int_{y\geq (-x)\vee 0} \bs e^\top_{k1}e^{\bs T_k (y+x)/\gamma}\bs e_{kk} \p(S_{k-1}\in \D y),\qquad x\in\R.\]

Next, we consider the cases $x\geq 0$ and $x<0$ separately and employ~\eqref{eq:GH} to find
\begin{align}\label{eq:Hx}
e^{\bs G x}\bs H(0)&=\frac{1}{\gamma}\sum_{k\geq 1}\int_{y\geq 0} \bs e^\top_{k1}e^{\bs T_k (y+x)/\gamma}\bs e_{kk} \p(S_{k-1}\in \D y),\qquad x\geq 0,\\
\bs H(-x)&=\frac{1}{\gamma}\sum_{k\geq 1}\int_{y> x} \bs e^\top_{k1}e^{\bs T_k (y-x)/\gamma}\bs e_{kk} \p(S_{k-1}\in \D y),\qquad x> 0.\nonumber
\end{align}
In the latter we use the fact that $y=x$ results in a zero matrix unless $k=1$, which can be disregarded 
%for negative levels 
because $S_0 = 0< x$.

We now show that the equality \eqref{eq:Hx} holds also for $x < 0$ by analytic continuation. First, $e^{\bs G x}\bs H(0)$ is a matrix of entire functions.  To see that the right-hand side of \eqref{eq:Hx} is  also a matrix of entire functions, we first write it as
\[
\frac 1 \gamma \sum_{k\geq 1}\bs e^\top_{k1}\e \big(e^{\bs T_k S_{k-1}/\gamma} \big)e^{\bs T_kx /\gamma }\bs e_{kk}. \] 
As in~\eqref{eq:normbound} we see that the maximal absolute entry of $\frac {\bs T_k} \gamma e^{\bs T_k x/\gamma}$ is upper bounded by $\frac r \gamma e^{r|x|/\gamma}$ for all $x\in \mathbb C,k\geq 1$.
By noting that $\e (e^{\bs T_k S_{k-1}/\gamma})$ has $[0,1]$ entries we get a bound
\[\sum_{k\geq 1}\left\|\bs e^\top_{k1}\e \big(e^{\bs T_k S_{k-1}/\gamma}\big)\frac {\bs T_k} \gamma e^{\bs T_kx/\gamma}\bs e_{kk}\right\|\leq \frac r \gamma e^{r|x|/\gamma} \sum_{k\geq 1}\left\|\bs e^\top_{k1}\e \big(e^{\bs T_k S_{k-1}/\gamma}\big)\bs e_{kk}\right\|
\leq r e^{r|x|/\gamma}  \sum_{i,j}\bs H(0)_{ij},\]
where in the latter step we upper bounded the matrix norm by the sum over all non-negative entries. 
In the defective case the matrix $\bs H(0)$ has finite entries~\cite[Lem.\ 10]{IP12}.
Now, according to, e.g.,~\cite[A16]{williams} differentiation at any $x\in\mathbb C$ can be performed under the summation sign, as desired.

We can now apply analytic continuation to find that~\eqref{eq:Hx} holds when $x\geq 0$ is replaced by $-x$.
Hence, \eqref{eq:Wrep} yields the stated expression of $\bs W(x)$ for $x>0$, whereas $\bs W(0)=\bs I_{\mathfrak{n}}/\gamma$ (see the comments following \cite[Eq.\ (14)]{IP12}) and so the formula is also true for $x=0$.

Finally, the non-defective case is obtained by a limit argument, by taking $q\downarrow 0$. 
%: kill $X$ at rate $q>0$ and then let $q\downarrow 0$. \kazu{should $q$ be a constant for all phases?}
It is known as in~the proof of \cite[Thm.\ 1]{IP12} that $\bs W(x)$ is continuous in $q$, and so it is left to take the limit inside the sum and integral of the stated expression. Finally, by using the bound in~\eqref{eq:normbound} to see  that the dominated convergence theorem applies, the proof is complete.
\end{proof}

For the case of deterministic jumps $C \equiv c>0$ as in the CUSUM case in Remark \ref{remark_SA_CUSUM}, the scale matrix can be written explicitly as a sum of matrix exponentials. In the next corollary, we also obtain  the integrated/differentiated scale matrices, which are required in Lemmas \ref{lemma_expectation} and \ref{lemma_expectation_dual}.  See Appendix \ref{remark_matrix_exp} for alternative expressions.
\begin{corollary} \label{example_const}
\rm
Suppose $C \equiv c>0$. We have
\begin{align}\label{eq:Wdeterministic}\bs W(x)&=\frac{1}{\gamma}\sum_{k=1}^{\lfloor x/c\rfloor+1}\bs e^\top_{k1}e^{\bs T_k (c(k-1)-x)/\gamma}\bs e_{kk}\qquad x \geq 0, \\
\bs W'_+(x) &=- \frac{1}{\gamma^2}\sum_{k=1}^{\lfloor x/c\rfloor+1}\bs e^\top_{k1} \bs T_k e^{\bs T_k (c(k-1)-x)/\gamma}\bs e_{kk}, \qquad x>0, \label{eq:Wdeterministic_prime} \\
\overline {\bs W}(x) &= \sum_{k=1}^{\lfloor x/c\rfloor+1}\bs e^\top_{k1}\bs T_k^{-1} \Big(\bs I_{nk}-e^{\bs T_k (c(k-1)-x)/\gamma}\Big)\bs e_{kk}, \qquad x\geq 0. \label{eq:Wdeterministic_integral}
\end{align}
%\red{except for $x$ divisible by $c$.}
\end{corollary}
\begin{proof}
The identity \eqref{eq:Wdeterministic} is a direct consequence of Theorem~\ref{thm:W} by using $S_{k-1} = c(k-1)$ for $k \geq 1$. The identities \eqref{eq:Wdeterministic_prime} and \eqref{eq:Wdeterministic_integral} can be derived by straightforward differentiation and integration of \eqref{eq:Wdeterministic}, where we also use Fubini's theorem for the latter.
\end{proof}

We conclude this section with other important examples. 
\begin{example}\rm 
\begin{enumerate}
\item Suppose the distribution of the interarrival times of $N$  is defective exponential of rate $\lambda>0$ killed at rate $q\geq 0$. 
In other words it is exponential of rate $\lambda+q$ which is declared killed with probability $q/(\lambda+q)$. 
Thus $\bs T=-\lambda-q,\bt=\lambda,\ba=1$ and we find that 
\[\bs e^\top_{k1}e^{\bs T_k x}\bs e_{kk}=\Big(\frac{\lambda}{\lambda+q}\Big)^{k-1} \frac {f_k(x;\lambda+q)} {\lambda+q},\qquad x > 0, k \geq 1, \]
where the function $ f_k(x;\xi) :=\xi^k x^{k-1}e^{-\xi x} /(k-1)!$ is the Erlang density. 
%In fact, we first observe this expression for $x>0$ using the probabilistic interpretation and then continue it to $x\in \R$. \kazu{ok to delete the previous sentence?}
Theorem~\ref{thm:W} now yields
\[
\bs W(x)=\frac{1}{\gamma}\sum_{k\geq 0}\frac{1}{k!}\Big(\frac{\lambda}{\lambda+q}\Big)^k\int_{y\in[0,x]} \big((y-x)(\lambda+q)/\gamma\big)^ke^{-(y-x)(\lambda+q)/\gamma} \p(S_k\in \D y)\qquad x\geq 0,\]
which coincides with the expression of the $q$-scale function of the L\'evy process (for the process killed at rate $q$) considered in~\cite[Thm.\ 2.2]{LW20}.

\item In Corollary \ref{example_const}, suppose additionally that the interarrival times of $N$ have non-defective Erlang distribution on $\mathfrak{n}$ phases with rate $\lambda$. Then the $(i,j)$-th entry of $\bs e^\top_{k1}e^{\bs T_k x}\bs e_{kk}$ is given by $f_{k\mathfrak{n}-(i-1)-(\mathfrak{n}-j)}(x;\lambda)/\lambda$,
which is understood as $0$ for $k=1, i>j$. Substituting this in \eqref{eq:Wdeterministic} gives
% \red{In particular, combining this expression with the previous Erlang inter-arrival times example we obtain 
 the formula in~\cite[Thm.\ 2]{AAI14}.
\end{enumerate}
\end{example}
%
%Our formula with the matrix exponential is very convenient, whereas writing it out leads to cumbersome expressions apart from a few examples considered above. 
% 
%It must be noted that Theorem~\ref{thm:W} can be extended to some other MAPs (evolving as a positive drift in each phase), but the formulas become considerably more complicated. \kazu{ok to delete this sentence?}
%
%
%%\subsection{Some exit formulas}
%
%%\[\int_0^z e^{-s x}\bs W(x)\D x=\sum_{k=1}^{\lfloor z/c\rfloor+1}\bs e^\top_{k1}(\gamma s\bs I_{nk}+\bs T_k)^{-1} \Big(e^{-s c(k-1)}\bs I_{nk}-e^{-s z}e^{\bs T_k (c(k-1)-z)/\gamma}\Big)\bs e_{kk}\qquad s,z\geq 0.\]
%%\jev{We may not need this $s$ which would simplify the formula somewhat}
%
%
%\jev{Get the pgf of the time to crossing the threshold after the switch.}
%This is easy conditional on the level $y$ at the switch-over time (see above) and will be expressed through $\overline{\bs W_1}(y)$. The level has a PH distribution - $d$ and in view of the above formula for 
%$\overline{\bs W_1}(y)$ the resultant expression should not be too bad. \kazu{ok to delete this paragraph?}

\section{Extension to the non-i.i.d.\ case with a change point} \label{section_non_iid}

We now generalize the results of Section \ref{sec_fluctuations_SA} to the non-i.i.d.\ case with a change point.  To this end, we introduce another (this time, non-defective discrete-time) Markov chain, which changes states immediately after each arrival. The change point $\nu$ is given by its first entry time to a certain closed set and is hence discrete-time PH distributed.  Furthermore, the distribution of the interarrival times is non-stationary and is modulated by this discrete-time Markov chain.  This generalization lets us analyze the CUSUM procedure when the change point $\nu$ is discrete-time PH and observations $\zeta$ (corresponding to the interarrival times) are non-i.i.d. As we did in 
Section \ref{sec_fluctuations_SA}, we first obtain first passage identities for the general case and then specialize them  for the analysis of CUSUM in Section \ref{subsection_CUSUM_noniid}.

More specifically, we let $Z = (Z_0,Z_1, \ldots)$ be a Markov chain on a state space $E^\nu_0 \cup E^\nu_1$ with $|E^\nu_0| = m_0$ and $|E^\nu_1| = m_1$ and label each state by
$E^\nu_0 = \{ (0, 1), \ldots, (0, m_0) \}$ and $E^\nu_1 = \{ (1, 1), \ldots, (1, m_1) \}$.  The sets $E^\nu_0$ and $E^\nu_1$ correspond, respectively, to the pre- and post-change states so that the change point is expressed as 
% \kazu{check}
%\[
%\nu = \sup (n \geq 0: Z_n \in E^\nu_0) \quad \textrm{and} \quad \nu+1 = \inf (n \geq 0: Z_n \in E^\nu_1),
%\]
%\red{
\begin{align}
 \nu = \inf (n \geq 0: Z_n \notin E^\nu_0) = \inf (n \geq 0: Z_n \in E^\nu_1). \label{nu_non_iid}
\end{align}
%where 
%%$\sup \varnothing = 0$ and 
%$\inf \varnothing = \infty$.
Necessarily $E^\nu_1$ is closed. We do not require $\nu < \infty$, but for the case this is certain $E^\nu_0$ is transient. 
We also allow $\nu = 0$, or equivalently $Z_0 \in E_1^\nu$, with a positive probability.
In particular, if $m_0 = 1$ (resp.\ $m_1 = 1$) then the observations are i.i.d.\  at or before (resp.\ after) $\nu$, conditionally given $\nu$.

Let the transition matrix and initial distribution of $Z$ be given by, respectively,
\begin{align} \label{transition_discrete}
\begin{pmatrix}
\bs {K} & \bs {L} \\ \bs {O} &  \bs {M} 
\end{pmatrix} \quad \textrm{and} \quad \bs{\beta} = (\bs{\beta}_1^{(0)}, \ldots, \bs{\beta}_{m_0}^{(0)}, \bs{\beta}_{1}^{(1)},  \ldots, \bs{\beta}^{(1)}_{m_1})
\end{align}
where $\bs {K}$ is $m_0 \times m_0$,  $\bs {L}$ is $m_0 \times m_1$, and  $\bs {M}$ is $m_1 \times m_1$. When $\nu = 0$ occurs with a positive probability, we have $\bs{\beta}_{l}^{(1)} > 0$ for some $1 \leq l \leq m_1$.

\begin{example}[robustness] \label{example_robust}
\rm
One of our motivations for considering this non-i.i.d.\ model is to provide a method to analyze the robustness of the CUSUM procedure. One typical approach for evaluating robustness is to consider the case where, for a certain (usually small) probability $\epsilon > 0$, the pre- and/or post-change distributions are different from the assumed $F_0$ and $F_1$ in the framework in Section \ref{subsection_change_point}.

As an illustration, suppose the post-change distribution is $F_1$ with probability $1-\epsilon$ and is $F_2$ with probability $\epsilon$. This can be modeled by setting $E_1^\nu = \{(1,1), (1,2) \}$ where on $(1,1)$ the observation is $F_1$-distributed whereas on $(1,2)$ it is $F_2$-distributed. Suppose further, for simplicity, that $\nu$ is zero-modified geometric
\[
\p (\nu = k)= \left\{ \begin{array}{ll} \mu, & k=0, \\ (1-\mu) (1-\lambda)^{k-1} \lambda, & k \geq 1, \end{array}\right. 
\]
for some $\mu \in [0,1)$ and $\lambda \in (0,1)$. 
We have $m_0 = 1$ and $m_1 = 2$ and the transition matrix and initial distribution of $Z$ as in  \eqref{transition_discrete} become, respectively,
\[
\begin{pmatrix}
1-\lambda & \lambda (1-\epsilon) & \lambda \epsilon \\
0 & 1 & 0 \\
0 & 0 & 1
\end{pmatrix} \quad \textrm{and} \quad (1-\mu, \mu (1-\epsilon), \mu \epsilon).
\]
Above, the distribution of $\nu$ is assumed to be independent of whether the post-change distribution is $F_1$ or $F_2$, but the case it is dependent can be also modeled by a simple modification; see the example given in our numerical results in Section \ref{subsection_example2}.
\end{example}
\vspace{0.3cm}
%Let $(N,J) = (N_t,J_t)_{t \geq 0}$ be  the Markovian arrival process, where $J$ is a Markov chain (with killing) with transition rate matrix $\bs T+\bt\ba$ 
%and the intensities $\bt\ba$ correspond to arrivals in the language of Markovian arrival processes~\cite[Ch.\ XI]{APQ}.
%In other words, it is a modification of the Markov chain 
%
% on \red{$ E \times (E^\nu_0 \cup E^\nu_1)$}

\begin{example}\rm Besides the geometric distribution, classical examples of discrete-time PH distributions include 
negative binomial and mixed geometric distributions, which can be realized by writing $\bs K$ and $\bs{\beta}$ in an obvious way (see, e.g., \cite{neuts1994matrix}).  In addition, $\mathrm{P}_k$ for $k \geq 1$ (where $\nu = k$ a.s.) in Section \ref{subsection_change_point}, which is of interest in the minimax formulation, can be modeled by using $k\times k$ matrix $\bs K$ with its entry  $1$ on the first diagonal above the main diagonal and $0$ otherwise and $\bs{\beta} = (1, 0, \ldots, 0)$.% \kazu{add dimensionof $X$ and the form of $\beta$}
\end{example}

We replace the background Markov chain $J$ considered in  Section \ref{sec_fluctuations_SA} with
a bivariate continuous-time Markov chain $(\tilde{J}, \tilde{Z}) = (\tilde{J}_t, \tilde{Z}_t)_{t \geq 0}$ defined as follows, and consider a new Markovian arrival process $\tilde{N}$. Here, $\tilde{Z}$ is a continuous-time Markov chain on $E^\nu_0 \cup E^\nu_1$ embedded by $Z$ with the law \eqref{transition_discrete} which changes states at each arrival so that
\begin{align}
\tilde{Z}_t = Z_{\tilde{N}_t}, \quad t \geq 0, \label{Z_tilde_Z}
\end{align}
where the evolution of the arrival process $\tilde{N}$  is modeled as follows.
Given $\tilde{Z} =(j,l) \in E^\nu_0 \cup E^\nu_1$,  the time until the next arrival is $\mathcal{PH}(E^{(j,l)}, \bs \alpha^{(j,l)}, \bs T^{(j,l)}, \bs t^{(j,l)})$-distributed. This is modeled by 
$\tilde{J}$ whose initial distribution is $\bs \alpha^{(j,l)}$ and transition rate matrix $\bs T^{(j,l)}$. % and arrival rate 
%$\bt^{(i,k)} \ba^{(i,k)}$.
The arrival occurs at rate $\bt^{(j,l)}$ and subsequently
% that happens  according to the intensity matrix $\bs t^{(i,k)} \ba^{(i,k)}$, 
 $\tilde{N}$ jumps up by one and $\tilde{Z}$ then changes its state according to its transition matrix given in \eqref{transition_discrete}. % We consider  the new Markov arrival process $N$ driven by this bivariate background Markov chain $(\tilde{J}, \tilde{Z})$.

In order to describe the law of the Markov chain  $(\tilde{J}, \tilde{Z})$, we label and order their states by
%of $(\tilde{J},\tilde{Z})$ by 
\begin{multline*}
\tilde{E}^\nu_0 \cup \tilde{E}^\nu_1 := 
\{ (1,(0,1)), \ldots, (\mathfrak{n}^{(0,1)},(0,1)), \ldots,  (1,(0,m_0)), \ldots, (\mathfrak{n}^{(0,m_0)},(0,m_0)),  \\ (1,(1,1)), \ldots, (\mathfrak{n}^{(1,1)},(1,1)), \ldots,  (1,(1,m_1)), \ldots, (\mathfrak{n}^{(1,m_1)},(1,m_1)) \},
\end{multline*}
where $\mathfrak{n}^{(i,l)} := |E^{(i,l)}|$ for $(i,l) \in E^\nu_0 \cup E^\nu_1$.
The size of the state space $\tilde{E}^\nu_0 \cup \tilde{E}^\nu_1$ is  $\tilde{\mathfrak{n}} := \sum_{l=1}^{m_0} \mathfrak{n}^{(0,l)} +  \sum_{l=1}^{m_1} \mathfrak{n}^{(1,l)}$. 
The initial distribution is given by $\tilde{\mathfrak{n}}$-dimensional row vector
\[
 \tilde{\ba}  := (\bs{\beta}_1^{(0)} \ba^{(0,1)},   \ldots, \bs{\beta}_{m_0}^{(0)} \ba^{(0,m_0)}, \bs{\beta}_{1}^{(1)} \ba^{(1,1)} ,   \ldots, \bs{\beta}_{m_1}^{(1)} \ba^{(1,m_1)} ).
\]
%Now we consider a continuous Markov chain with states $(i,k) \in E \times (E^\theta \cup \Delta^\theta)$. If $k \in E^\theta$, this means the change point has not occurred, while If $k \in \Delta\theta$, this means the change has occurred. The state $k$ changes only at the arrival which comes at rate $\bs t t^{(0)} \alpha^{(1)}$ 
The transition intensity matrix is an $\tilde{\mathfrak{n}} \times \tilde{\mathfrak{n}}$ matrix $\tilde{\bs{T}} +\tilde{\bs B}$ where
\[
\tilde{\bs{T}} := \textrm{diag}(\bs{T}^{(0,1)}, \ldots, \bs{T}^{(0,m_0)}, \bs{T}^{(1,1)}, \ldots, \bs{T}^{(1,m_1)})
\]
is the non-arrival intensities whereas
\[
\tilde{\bs B} :=  \begin{pmatrix} \mathcal{K} & \mathcal{L} \\ \bs O & \mathcal{M}
\end{pmatrix}
% \begin{pmatrix} \bs{K}_{11}\bs t^{(0,1)}  \alpha^{(0,1)}   & \ldots & \bs{K}_{1m_0} \bs t^{(0,1)}  \alpha^{(0,m_0)}  & \bs{L}_{11}\bs t^{(0,1)}  \alpha^{(1,1)}   & \ldots & \bs{L}_{1m_1} \bs t^{(0,1)}  \alpha^{(1,m_1)}  \\  \vdots  & \ddots & \vdots & \vdots & \ddots & \vdots  \\ \bs{K}_{m_01}\bs t^{(0,m_0)}  \alpha^{(0,1)}  & \ldots & \bs{K}_{m_0m_0} \bs t^{(0,m_0)}  \alpha^{(0, m_0)}  & \bs{L}_{m_01}\bs t^{(0,m_0)}  \alpha^{(1,1)}    & \ldots & \bs{L}_{m_0m_1} \bs t^{(0,m_0)}  \alpha^{(1,m_1)}   \\ \bs 0 & \ldots & \bs 0   & \bs{M}_{11}\bs t^{(1,1)}  \alpha^{(1,1)}    & \ldots & \bs{M}_{1 m_1} \bs t^{(1,1)}  \alpha^{(1,m_1)}   
%\\  \vdots  & \ddots & \vdots & \vdots & \ddots & \vdots  \\ 
%\bs 0 & \ldots & \bs 0   & \bs{M}_{m_1 1}\bs t^{(1,m_1)}  \alpha^{(1,1)}    & \ldots & \bs{M}_{m_1 m_1} \bs t^{(1,m_1)}  \alpha^{(1,m_1)}  
%\end{pmatrix}
\]
with
\begin{align*}
\mathcal{K} &:= \begin{pmatrix} \bs{K}_{11}\bs t^{(0,1)}  \ba^{(0,1)}   & \ldots & \bs{K}_{1m_0} \bs t^{(0,1)}  \ba^{(0,m_0)}   \\  \vdots  & \ddots & \vdots   \\ \bs{K}_{m_01}\bs t^{(0,m_0)}  \ba^{(0,1)}  & \ldots & \bs{K}_{m_0m_0} \bs t^{(0,m_0)}  \ba^{(0, m_0)} 
\end{pmatrix}, \\
\mathcal{L} &:=  \begin{pmatrix} \bs{L}_{11}\bs t^{(0,1)}  \ba^{(1,1)}   & \ldots & \bs{L}_{1m_1} \bs t^{(0,1)}  \ba^{(1,m_1)}  \\   \vdots & \ddots & \vdots  \\    \bs{L}_{m_01}\bs t^{(0,m_0)}  \ba^{(1,1)}    & \ldots & \bs{L}_{m_0m_1} \bs t^{(0,m_0)}  \ba^{(1,m_1)} 
\end{pmatrix}, \\
\mathcal{M}  &:=  \begin{pmatrix}   \bs{M}_{11}\bs t^{(1,1)}  \ba^{(1,1)}    & \ldots & \bs{M}_{1 m_1} \bs t^{(1,1)}  \ba^{(1,m_1)}   
\\   \vdots & \ddots & \vdots  \\ 
 \bs{M}_{m_1 1}\bs t^{(1,m_1)}  \ba^{(1,1)}    & \ldots & \bs{M}_{m_1 m_1} \bs t^{(1,m_1)}  \ba^{(1,m_1)}  
\end{pmatrix},
\end{align*}
is the arrival intensities. %Notice that $\tilde{\bs{T}}$ and $\bs B$ are $[n(m+1)]$ by $[n(m+1)]$ matrices.

\begin{example} \label{example_geometric}
\rm 
 In the setting of Example \ref{example_robust} where $F_0 \sim \mathcal{PH}(E^{(0,1)}, \bs \alpha^{(0,1)}, \bs T^{(0,1)}, \bs t^{(0,1)})$ and \\
 $F_l \sim \mathcal{PH}(E^{(1,l)}, \bs \alpha^{(1,l)}, \bs T^{(1,l)}, \bs t^{(1,l)})$ for $l = 1,2$,
 % and $F_2 \sim \mathcal{PH}(E^{(1,2)}, \bs \alpha^{(1,2)}, \bs T^{(1,2)}, \bs t^{(1,2)})$,
%For the case $m_0=m_1 = 1$ and the change-point $\nu$ is zero-modified geometric
%\[
%\p (\nu = k)= \left\{ \begin{array}{ll} \mu, & k=0, \\ (1-\mu) (1-\lambda)^{k-1} \lambda, & k > 0, \end{array}\right. 
%\]
we have
\[
 \tilde{\ba}  = ( (1-\mu)\ba^{(0,1)},  \mu (1-\epsilon)\ba^{(1,1)}, \mu \epsilon \ba^{(1,2)} )
\]
and
%we consider expanding the state space to $\{ (1,0), \ldots, (n,0),  (1,1), \ldots, (n,1)\}$ where the second variable indicates ``pre-change'' if it is $0$ and ``post-change'' if it is $1$.  It changes $0$ to $1$ at arrival times with probability $\lambda$. Then,
%we can replace $\bs T$ and $\bt\ba$, respectively, with 
\[
\tilde{\bs{T}}  = \begin{pmatrix} \bs T^{(0,1)} & \bs O  & \bs O \\ \bs O & \bs T^{(1,1)}  & \bs O \\  \bs O  & \bs O  &  \bs T^{(1,2)}  \end{pmatrix}, \quad 
\tilde{\bs B}  =  \begin{pmatrix} (1-\lambda) \bs t^{(0,1)} \bs{\alpha}^{(0,1)}  & \lambda (1-\epsilon)\bs t^{(0,1)}  \bs{\alpha}^{(1,1)}  &   \lambda \epsilon \bs t^{(0,1)}  \bs{\alpha}^{(1,2)}  \\ \bs O & \bs t^{(1,1)}  \bs{\alpha}^{(1,1)} &  \bs O   \\ \bs O &  \bs O   & \bs t^{(1,2)}  \bs{\alpha}^{(1,2)}  \end{pmatrix}.
\]
\end{example}

We now consider the MAP $(X, (\tilde{J},\tilde{Z}))$ given by
\begin{equation}\label{eq:MAP_general}
X_t:= X_0 + \gamma t-S_{\tilde{N}_t}, \qquad t \geq 0, \end{equation}
with the same $S$  as in \eqref{def_S}, as a generalization of  \eqref{eq:MAP}.  Because the only change made from  \eqref{eq:MAP} is the modulator of the MAP, whose law is completely specified by its transition rate matrix and its initial distribution, it is  clear that Theorem \ref{thm:W} holds by simply replacing  $\bs T$ and $\bs B$ with $\tilde{\bs{T}}$ and $\tilde{\bs B}$, respectively.  Hence, we have the following.

\begin{theorem} \label{theorem_scale_matrix_general}
The scale matrix of a MAP $(X, (\tilde{J},\tilde{Z}))$ as in \eqref{eq:MAP_general} has the representation
\[\bs W(x)=\frac{1}{\gamma}\sum_{k\geq 1}\int_{y\in[0,x]} \tilde{\bs e}^\top_{k1}e^{\tilde{\bs T}_k (y-x)/\gamma} \tilde{\bs e}_{kk} \p(S_{k-1}\in \D y), \qquad x\geq 0,\]
where $\tilde{\bs T}_k$, $\tilde{\bs e}_{k1}$, and $\tilde{\bs e}_{kk}$ are as in \eqref{def_T_k} with $\bs T$, $\bs B$, and $\bs I_{\mathfrak{n}}$ replaced with $\tilde{\bs T}$, $\tilde{\bs B}$, and $\bs I_{\tilde{\mathfrak{n}}}$, respectively, for all $k \geq 1$.
%\[\tilde{\bs T}_k=\begin{pmatrix}
%\tilde{\bs{T}}& \tilde{\bs B}  &\ldots&\bs O\\
% &\ddots&\ddots&\vdots\\
%\vdots& & \tilde{\bs{T}} &\tilde{\bs B}  \\
%\bs O&\ldots&  & \tilde{\bs{T}}
%\end{pmatrix},
%\qquad \tilde{\bs e}_{k1}=\begin{pmatrix}
%\bs I_{\tilde{n}}\\ \vdots\\ \bs O\\ \bs O
%\end{pmatrix},
%\qquad \tilde{\bs e}_{kk}=\begin{pmatrix}
%\bs O \\ \vdots\\\bs O\\ \bs I_{\tilde{n}}
%\end{pmatrix}, \quad k \geq 1. \]
\end{theorem}

Using this generalized scale matrix in Theorem \ref{theorem_scale_matrix_general}, the identity \eqref{J_at_passage} and Lemma \ref{lemma_expectation} can be extended as follows. Below, it is understood that  the first passage time $\tau_a$ as in \eqref{tau_a_def} is for the generalized $X$ defined in \eqref{eq:MAP_general}.  In view of \eqref{nu_non_iid} and \eqref{Z_tilde_Z}, 
\begin{align}
\{ \tilde{N}_{\tau_a } < \nu \} =  \{ Z_1, \ldots, Z_{\tilde{N}_{\tau_a}} \in E_0^\nu \} = \{ Z_{\tilde{N}_{\tau_a}} \in E_0^\nu \}  =  \{ \tilde{Z}_{\tau_a } \in E_0^\nu \} \label{FA_identity}
\end{align}
where in the second equality, we use that $E_1^\nu$ is closed.
We let the $\tilde{\mathfrak{n}}$-dimensional column vectors
\begin{align*}
\tilde{\bt}^{(0)} &= [(\bt^{(0,1)})^\top,   \ldots, (\bt^{(0,m_0)})^\top, \bs 0^\top,   \ldots, \bs 0^\top ]^\top, \\
\tilde{\bt}^{(1)} &= [ \bs 0^\top,   \ldots, \bs 0^\top, (\bt^{(1,1)})^\top,   \ldots, (\bt^{(1,m_1)})^\top ]^\top, \\
\tilde{\bt} &= \tilde{\bt}^{(0)}  + \tilde{\bt}^{(1)}  = [(\bt^{(0,1)})^\top,   \ldots, (\bt^{(0,m_0)})^\top, (\bt^{(1,1)})^\top,   \ldots, (\bt^{(1,m_1)})^\top ]^\top,
\end{align*}
be the rate of arrivals coming from $\tilde{E}_0^\nu$ and $\tilde{E}_1^\nu$ and their sum. We also let
\begin{align*}
\tilde{\bs 1}^{(0)} = [\bs 1^\top,  \ldots, \bs  1^\top, \bs  0^\top, \ldots, \bs  0^\top ]^\top \quad \textrm{and} \quad
\tilde{\bs 1}^{(1)} = [\bs 0^\top,  \ldots, \bs  0^\top, \bs  1^\top, \ldots, \bs  1^\top ]^\top,
\end{align*}
whose element is $1$ or $0$ depending on whether it belongs to $\tilde{E}_0^\nu$ or $\tilde{E}_1^\nu$.

As for \eqref{N_sub}, we let $\tilde{N}_{\tau_a}(k) := \sum_{t \leq \tau_a: \Delta \tilde{N}_t \neq 0} 1_{\{ (\tilde{J}_{t-}, \tilde{Z}_{t-}) = k \}}$ for $a > 0$ and $k \in \tilde{E}_0^\nu \cup \tilde{E}_1^\nu$.

\begin{lemma} \label{lemma_non-iid}
(1) 
For $a > 0$, we have
\[\p [(\tilde{J},\tilde{Z})_{\tau_a} ]=  \Big(\bs I_{\tilde{\mathfrak{n}}}-\overline {\bs W}(a)(\tilde{\bs{T}}+ \tilde{\bs B} )\Big)^{-1}. \] 
Hence, by \eqref{FA_identity},
\begin{align*}
\mathbb{P} (\tilde{N}_{\tau_a } < \nu   |  (\tilde{J}_0,\tilde{Z}_0) \sim \tilde{\ba}) =
\mathbb{P} (\tilde{Z}_{\tau_a } \in E_0^\nu  |  (\tilde{J}_0,\tilde{Z}_0) \sim \tilde{\ba}) =
%\p^{\tilde{\ba}} (N_{\tau_a } < \nu ) 
%= \p^{\tilde{\ba}} (\tilde{Z}_{\tau_a } \in E_0^v)  &= 
\tilde{\ba} \Big(\bs I_{\tilde{\mathfrak{n}}}-\overline {\bs W}(a)(\tilde{\bs{T}}+ \tilde{\bs B} )\Big)^{-1} 
\tilde{\bs 1}^{(0)}.
%[\bs 1^\top,  \ldots, \bs  1^\top, \bs  0^\top, \ldots, \bs  0^\top ]^\top, \\
%\p^{\tilde{\ba}} (N_{\tau_x } \geq \nu ) = \p^{\tilde{\ba}} (\tilde{Z}_{\tau_x } \in E_1^v) &= 
%\tilde{\ba} \Big(\bs I_{\tilde{n}}-\overline {\bs W}(x)(\tilde{\bs{T}}+ \bs B)\Big)^{-1} 
%\tilde{\bs 1}^{(1)}. 
%[\bs 0^\top,  \ldots, \bs  0^\top, \bs  1^\top, \ldots, \bs  1^\top ]^\top. 
 \end{align*}
(2)   Suppose $(\tilde{J}, \tilde{Z})$ is non-defective (i.e. $\bs q = \bs 0$).
%\red{Suppose the MAP is not killed (i.e. $\bs q^{(i,k)} = \bs 0$).}
For $a > 0$, we have
%\[\e [N_{\tau_x}; (\tilde{J},\tilde{Z})_{\tau_x} ]= \Big(\bs I_{\tilde{n}}-\overline {\bs W}(x)(\tilde{\bs T}+\bs B)\Big)^{-1}\overline {\bs W}(x) \mathrm{diag} (\bt^{(0,1)}, \ldots, \bt^{(0,m_0)}, \bt^{(1,1)},  \ldots, \bt^{(1,m_1)}).\]
%\red{
%\[\e_{0,i} (N_{\tau_x}(j) ) _{ij} = \Big(\bs I_{\tilde{n}}-\overline {\bs W}(x)(\tilde{\bs T}+\bs B)\Big)^{-1}\overline {\bs W}(x) \mathrm{diag} (\bt^{(0,1)}, \ldots, \bt^{(0,m_0)}, \bt^{(1,1)},  \ldots, \bt^{(1,m_1)}).\]
%}
\[\e_{0,i} (\tilde{N}_{\tau_a}(k)) = \Big[ \Big(\bs I_{\tilde{\mathfrak{n}}} - \overline{\bs W}(a) (\tilde{\bs T}+\tilde{\bs B} )\Big)^{-1}  \overline{\bs W}(a)  \mathrm{diag} (\tilde{\bt}) \Big]_{ik}, \quad i,k \in \tilde{E}_0^\nu \cup \tilde{E}_1^\nu.
\]
Hence, %the unconditional expectations of  $N_{\tau_x}$ and $(N_{\tau_x} - \nu)^+$  when $(J_0,Z_0) \sim \tilde{\ba}$ are
\begin{align*}
\mathbb{E} (\tilde{N}_{\tau_a}  | (\tilde{J}_0,\tilde{Z}_0) \sim \tilde{\ba}) 
%\e^{\tilde{\ba}} (N_{\tau_a}) 
&= \tilde{\ba}  \Big(\bs I_{\tilde{\mathfrak{n}}}-\overline {\bs W}(a)(\tilde{\bs T}+\tilde{\bs B} )\Big)^{-1}\overline {\bs W}(a) \tilde{\bt}, \\
%[(\bt^{(0,1)})^\top,   \ldots, (\bt^{(0,m_0)}), (\bt^{(1,1)})^\top,   \ldots, (\bt^{(1,m_1)})^\top ]^\top, \\
\mathbb{E} \Big(\sum_{t \leq \tau_a: \Delta \tilde{N}_t \neq 0} 1_{\{ \tilde{Z}_{t-} \in E_1^\nu \}}\Big|  (\tilde{J}_0,\tilde{Z}_0) \sim \tilde{\ba} \Big) 
%\e^{\tilde{\ba}} \Big(\sum_{t \leq \tau_a: \Delta N_t \neq 0} 1_{\{ \tilde{Z}_{t-} \in E_1^\nu \}} \Big) 
&= 
\tilde{\ba}  \Big(\bs I_{\tilde{\mathfrak{n}}}-\overline {\bs W}(a)(\tilde{\bs T}+\tilde{\bs B} )\Big)^{-1}\overline {\bs W}(a) 
\tilde{\bt}^{(1)}.
%[\bs 0^\top,   \ldots, \bs 0^\top,  (\bt^{(1,1)})^\top,   \ldots, (\bt^{(1,m_1)})^\top  ]^\top.
\end{align*}
%With the addition of one for the detection delay
%\[
%\e ((N_{\tau_x} - \theta)^+) \red{+ \alpha} = 
%\tilde{\ba}  \Big(\bs I_{\tilde{n}}-\overline {\bs W}(x)(\tilde{\bs{T}}+ \bs B)\Big)^{-1} \Big( \overline {\bs W}(x)\begin{pmatrix}
%\bs 0 \\ \bs t_1
%\end{pmatrix} + \begin{pmatrix}
%\bs 0 \\ \bs 1
%\end{pmatrix} \Big). \]
\end{lemma}

%
%
%\kazu{For the geometric change point case, just replace $\bs T$ with $\tilde{\bs{T}}$ and $\bt\ba$ with $\bs B$ and note that
%\[
%\frac \partial {\partial z} F(0) \bs 1 = \bs B \bs 1 = \begin{pmatrix}
%\bs t_0 \\ \bs t_1
%\end{pmatrix}.
%\]
%For detection delay, only kill when it is in the post-change state. Hence,
%\[F(s):=\theta s\bs I+\tilde{\bs{T}} + \begin{pmatrix} (1-\lambda) \bs t_0 \alpha_0 & \lambda \bs t_0 \alpha_1 \\ \bs 0 & z \bs t_1 \alpha_1 \end{pmatrix} e^{-\kappa(\theta)s}. \]
%In this case,
%\[
%\frac \partial {\partial z} F(0) \bs 1 = \begin{pmatrix}
%\bs 0 \\ \bs t_1
%\end{pmatrix}.
%\]
%Regarding the addition by one, it should be counted only when it is post-change at the upcrossing time. Hence, the detection delay after addition is
%\[[(1-\mu)\bs \alpha_0, \mu \bs \alpha_1 ]\Big(\bs I_{2n}-\overline {\bs W}(x)(\tilde{\bs{T}}+ \bs B)\Big)^{-1} \Big( \overline {\bs W}(x)\begin{pmatrix}
%\bs 0 \\ \bs t_1
%\end{pmatrix} + \begin{pmatrix}
%\bs 0 \\ \bs 1
%\end{pmatrix} \Big). \]
%}
%
%\kazu{False alarm is the event $\{\theta > T\}$ that it is stopped before the change. Its probability is an important quantity and it is given by
%\[[(1-\mu)\bs \alpha_0, \mu \bs \alpha_1 ]\Big(\bs I_{2n}-\overline {\bs W}(x)(\tilde{\bs{T}}+ \bs B)\Big)^{-1} \begin{pmatrix}
%\bs 1 \\ \bs 0
%\end{pmatrix}. \]
%}

\subsection{Spectrally positive case}
Suppose  temporarily that
\begin{align}
(\widetilde{\textrm{SP}}) \qquad X_t = X_0 - \gamma t + S_{\tilde{N}_t}, \quad t \geq 0.  \label{SP_noniid}
\end{align}
 We use the same notations as those in Section  \ref{subsection_SP}, except that we replace $J$  by $(\tilde{J}, \tilde{Z})$ and $N$ by $\tilde{N}$.

It is clear that \eqref{first_passage_SP} immediately gives, for this generalized case:
\begin{align}
\p_{x} [(\tilde{J},\tilde{Z})_{\tau_a}] =  \bs I_{\tilde{\mathfrak{n}}} - \Big( \overline{\bs W}^d(a-x)  - \bs W^d(a-x)
(\bs W^{d})'_+(a)^{-1}
{\bs W}^d(a) \Big) (\tilde{\bs T}+\tilde{\bs B} ), \quad a > 0, \; 0 \leq x \leq a. \label{first_passage_SP_noniid}
\end{align}
Likewise, by Lemma \ref{lemma_expectation_dual}, we have the following.
\begin{lemma}  \label{lemma_SP_non_iid}
Suppose  ($\widetilde{\textrm{SP}}$).
For $a > 0$ and $0 \leq x \leq a$,
%we have
%\[\e_y [N_{\tilde{\tau}_x}; J_{\tau_x} ] =- \Big(\int_0^{x-y} \tilde {\bs W}(z)\D z-\tilde {\bs W}(x-y)\tilde {\bs W}'_+(x)^{-1}\tilde {\bs W}(x)\Big) \mathrm{diag} (\bt).\]
\[\e_{x,i} (\tilde{N}_{\tau_a}(k)) =-  \Big[\Big( \overline{\bs W}^d(a-x)- \bs W^d(a-x) (\bs W^{d})'_+(a)^{-1} {\bs W}^d(a)\Big) \mathrm{diag}(\bt) \Big]_{ik}, \quad i, k \in E, \]
%where $N_{\tau_x}(j)$ which is the number of arrivals occurring in phase $j$ until $\tau_x$ [check]
and hence
%Hence, %the unconditional expectation when $J_0 \sim \ba$ is
\begin{align*}
\mathbb{E}_x( \tilde{N}_{\tau_a}   | (\tilde{J}_0,\tilde{Z}_0) \sim \tilde{\ba}) 
%=\e_y^{\tilde{\ba}} (N_{\tau_a}  ) 
&=- \tilde{\ba} \Big( \overline{\bs W}^d(a-x)- \bs W^d(a-x) (\bs W^{d})'_+(a)^{-1} {\bs W}^d(a)\Big) \tilde{\bt}, \\
\mathbb{E}_x \Big(\sum_{t \leq \tau_a: \Delta \tilde{N}_t \neq 0} 1_{\{ \tilde{Z}_{t-} \in E_1^\nu \}} \Big| (\tilde{J}_0,\tilde{Z}_0) \sim \tilde{\ba} \Big) 
%\e^{\tilde{\ba}} \Big(\sum_{t \leq \tau_a: \Delta N_t \neq 0} 1_{\{ \tilde{Z}_{t-} \in E_1^\nu \}} \Big) 
&=- \tilde{\ba} \Big( \overline{\bs W}^d(a-x)- \bs W^d(a-x) (\bs W^{d})'_+(a)^{-1} {\bs W}^d(a)\Big) \tilde{\bt}^{(1)}.
\end{align*}
\end{lemma}

\subsection{The case of CUSUM} \label{subsection_CUSUM_noniid}

With the above results for the non-i.i.d.\ case, more interesting quantities can be computed beyond those obtained in Section \ref{subsection_SA_CUSUM}.
Here, we compute the average detection delay \eqref{def_ADD} and false alarm probability \eqref{def_PFA}, in addition to the average run length \eqref{def_ARL}.
%Again recall our discussions in Section \ref{subsection_change_point}. 

Below, we consider the change point $\nu$ and observation $\zeta$, modeled by $(\tilde{J}, \tilde{Z})$. As we did in Section \ref{subsection_SA_CUSUM}, we shall first consider the case $\theta > 0$ and then the case $\theta < 0$.

%is discrete-time PH  distributed and the observation distributions used in this section, which are completely described by 

%Now recall our discussions in Section \ref{subsection_change_point}. The optimal barrier in the minimax formulation is given by $A_{\beta}$ such that \eqref{lorden_beta} holds. For this computation we need the average run length $\mathrm{ARL}(T_A) \equiv \mathrm{E}_\infty [T_A]$ for every level $A > 0$. Here, we consider the case $F_0 \sim \mathcal{PH}(E, \bs \alpha, \bs T, \bs t)$; note that any positive distribution can be approximated by a phse-type distribution and hence this assumption is not restrictive. 
%%
%
% where under the law $\mathrm{P}_\infty$ corresponds to the case discussed in
%Remark \ref{remark_SA_CUSUM} with $F_0$ given by  
%the observation is independent and $F_0$-distributed.
\subsubsection{For the case $\theta > 0$}
 Let $\bs W$ be the scale matrix of the MAP $(X, (\tilde{J},\tilde{Z}))$ as in \eqref{eq:MAP_general} with $\gamma = \theta$ and  $C \equiv \kappa(\theta) > 0$ (see Remark \ref{remark_SA_CUSUM}).

%An important application in the analysis of the CUSUM strategy is for the computation of the probability of false  alarm $\{T_A \leq \nu \}$  as well as the expectation of the detection delay $(T_A - \nu)^+$. 
%Recall Proposition \ref{proposition_equivalence}(1)  for the link between $T_A$ and the first-passage time of $Y$.
% that  $T_A
%= 1 + N_{\tau_{A+ \kappa(\lambda)}}$.
%, which is confirmed in Proposition \ref{proposition_equivalence}(1).
% analysis of the CUSUM strategy as discussed in Section TODO.
\begin{corollary} \label{cor_CUSUM_noniid}
%Suppose  $(T_A, \mathrm{P})  = ( 1 + N_{\tau_{A+ \kappa(\theta)}},  \p^{\tilde{\ba}})$.
%Recall Proposition \ref{proposition_equivalence}.
Fix $A > 0$.
(1) We have
\begin{align*}
%\p^{\tilde{\ba}} 
\mathrm{PFA}( T_A )  &=  \mathbb{P} \Big(\tilde{N}_{\tau_{A+\kappa(\theta)}} +1 \leq \nu | (\tilde{J}_0,\tilde{Z}_0) \sim \tilde{\ba} \Big) = 
\tilde{\ba} \Big(\bs I_{\tilde{\mathfrak{n}}}-\overline {\bs W}(A+\kappa(\theta))(\tilde{\bs{T}}+ \tilde{\bs B} )\Big)^{-1} 
\tilde{\bs 1}^{(0)}.
%[\bs 1^\top,  \ldots, \bs  1^\top, \bs  0^\top, \ldots, \bs  0^\top ]^\top, \\
%\p^{\tilde{\ba}}
%\mathrm{P} ( T_A > \nu )  &= 
%\tilde{\ba} \Big(\bs I_{\tilde{n}}-\overline {\bs W}(A+\kappa(\lambda))(\tilde{\bs{T}}+ \bs B)\Big)^{-1} 
%\tilde{\bs 1}^{(1)}. 
%[\bs  0^\top, \ldots, \bs  0^\top, \bs 1^\top,  \ldots, \bs  1^\top ]^\top.
\end{align*}
%For the latter, by conditioning at the hitting time ???
%\begin{align*}
%\p  \{ N_{\tau^{(|\kappa(\lambda)|)}_{A+ |\kappa(\lambda)|}} = \nu \} = \p [(J,Z)_{\tau_x} ] [\bs  1^\top, \ldots   \bs  1^\top, \bs 0^\top,  \ldots, \bs  0^\top  ]^\top \times TODO
% \end{align*}
(2) We have
\begin{align*}
%\e^{\tilde{\ba}} 
\mathrm{ARL}(T_A ) &=    1 + \e (\tilde{N}_{\tau_{A+ \kappa(\theta)}} |  (\tilde{J}_0,\tilde{Z}_0 ) \sim \tilde{\ba}  ) \\ &= 
%1+ \e^{\tilde{\ba}} (N_{\tau_{A+\kappa(\lambda)}} )=
 1+ \tilde{\ba} \Big(\bs I_{\tilde{\mathfrak{n}}}-\overline {\bs W}(A+\kappa(\theta))(\tilde{\bs{T}}+\tilde{\bs B} )\Big)^{-1}\overline {\bs W}(A+\kappa(\theta)) \tilde{\bt}, \\
%\e^{\tilde{\ba}} 
\mathrm{ADD}(T_A ) &= \e \big( (\tilde{N}_{\tau_{A+\kappa(\theta)} } + 1 - \nu)^+  |  (\tilde{J}_0,\tilde{Z}_0 ) \sim \tilde{\ba}  \big) \\ &= 
%&= \p (N_{\tau_{A+\kappa(\lambda)} } \geq \nu ) + \e ((N_{\tau_{A+\kappa(\lambda)}} - \nu)^+ ) \\
\tilde{\ba}  \Big(\bs I_{\tilde{\mathfrak{n}}}-\overline {\bs W}(A+\kappa(\theta))(\tilde{\bs T}+\tilde{\bs B} )\Big)^{-1}  \Big( \overline {\bs W}(A+\kappa(\theta))
\tilde{\bt}^{(1)} + \tilde{\bs 1}^{(1)} 
 %[\bs 0^\top,   \ldots, \bs 0^\top,  (\bt^{(m_0+1)})^\top, \ldots (\bt^{(m_0+m_1)})^\top ]^\top 
 %+ [\bs  0^\top, \ldots, \bs  0^\top, \bs 1^\top,  \ldots, \bs  1^\top ]^\top 
 \Big).
\end{align*}
\end{corollary}
\begin{proof}
(1)  By Proposition \ref{proposition_equivalence}(1)   and \eqref{FA_identity},
$\{ T_A \leq \nu \} = \{ \tilde{N}_{\tau_{A+\kappa(\theta)}} +1 \leq \nu \} =  \{ \tilde{N}_{\tau_{A+\kappa(\theta)}} < \nu \} = \{ \tilde{Z}_{\tau_{A+\kappa(\theta)}} \in E_0^\nu \}$.
Hence, Lemma \ref{lemma_non-iid}(1) gives the results.

(2) The first claim is immediate by  Proposition \ref{proposition_equivalence}(1)    and Lemma \ref{lemma_non-iid}(2).
% gives
%\[\e^{\tilde{\ba}} (T_A ) = 1+ \e^{\tilde{\ba}} (N_{\tau_{A+\kappa(\lambda)}} )= 1+ \tilde{\ba} \Big(\bs I_{\tilde{n}}-\overline {\bs W}(A+\kappa(\lambda))(\tilde{\bs{T}}+\bs B)\Big)^{-1}\overline {\bs W}(A+\kappa(\lambda)) \tilde{\bt}.
%%[(\bt^{(1)})^\top,   \ldots, (\bt^{(m_0+m_1)})^\top ]^\top.
%\]
%\kazu{define $\tilde{\bt}$}
%Note that
%\[
%  \sum_{t \leq \tau_x: \Delta N_t \neq 0} 1_{\{ \tilde{Z}_{t-} \in E_1^\nu \}} = \sum_{k \in N_1^\nu} N_{\tau_x}(k)
%%\\
%= 
%%\tilde{\ba} \Big(I- \overline{\bs W}_1(x) (\tilde{T}+ B)\Big)^{-1}  \overline{\bs W}_1(x) \tilde{\bt}^{(1)}.
%\tilde{\ba}  \Big(\bs I_{\tilde{n}}-\overline {\bs W}(x)(\tilde{\bs T}+\bs B)\Big)^{-1}\overline {\bs W}(x) 
%\tilde{\bt}^{(1)}.
%%\sum_{j} \mathcal{M}_{ij} \alpha^{(1,j)}.
%\]
For the second claim, 
\begin{align*}
(T_A - \nu)^+ = (\tilde{N}_{\tau_{A+\kappa(\theta)} } + 1 - \nu)^+ &=  \Big( \sum_{t \leq \tau_{A+\kappa(\theta)}: \Delta \tilde{N}_t \neq 0} 1_{\{ \tilde{Z}_{t-} \in E_1^\nu \}} \Big) + 1_{\{ \tilde{Z}_{\tau_A+\kappa(\theta)}\in E_1^\nu \}}.
%&=  \sum_{t \leq \tau_{A+\kappa(\lambda)}: \Delta N_t \neq 0} 1_{\{ \tilde{Z}_{t-} \in E_1^\nu \}} + 1_{\{N_{\tau_{A+\kappa(\lambda)} } \geq \nu \}}.
\end{align*}
Hence, Lemma \ref{lemma_non-iid}(1) and (2) give the result.
%\begin{align*}
%\e^{\tilde{\ba}} ((T_A - \nu)^+ ) 
%%&= \p (N_{\tau_{A+\kappa(\lambda)} } \geq \nu ) + \e ((N_{\tau_{A+\kappa(\lambda)}} - \nu)^+ ) \\
%&= 
%\tilde{\ba}  \Big(\bs I_{\tilde{n}}-\overline {\bs W}(A+\kappa(\lambda))(\tilde{\bs T}+\bs B)\Big)^{-1}  \Big( \overline {\bs W}(A+\kappa(\lambda))
%\tilde{\bt}^{(1)} + \tilde{\bs 1}^{(1)} 
% %[\bs 0^\top,   \ldots, \bs 0^\top,  (\bt^{(m_0+1)})^\top, \ldots (\bt^{(m_0+m_1)})^\top ]^\top 
% %+ [\bs  0^\top, \ldots, \bs  0^\top, \bs 1^\top,  \ldots, \bs  1^\top ]^\top 
% \Big).
%\end{align*}
\end{proof}

%\begin{remark}
% \kazu{
% For example when $E_1^\nu = E_{1,a}^\nu \cup  E_{1,b}^\nu$ both of which is closed then we can obtain, e.g., $\p (N_{\tau_x } \geq \nu, A ) = \p (\tilde{Z}_{\tau_x } \in E_{1,a}^v)$. 
% }
% 
% \red{We can also obtain for any $E_0$
%\[
%\sum 1_{\{ k > \nu, Z_k \in E_0 \}}
%\]
%}
%\end{remark}

\begin{remark} 
Notice that more variations can be computed. For example, in Examples  \ref{example_robust} and  \ref{example_geometric}, one can for example compute
$\mathrm{P} ( T_A > \nu,  H_l)$ and $\mathrm{E} ((T_A - \nu)^+ ; H_l)$ where $H_l$ is the event that the true post-change distribution is $F_l$ for $l = 1,2$. Indeed, 
\begin{align*}
\{ T_A > \nu, H_l \} = \{ \tilde{N}_{\tau_{A+\kappa(\theta)}} +1 > \nu, H_l \} =  \{ \tilde{N}_{\tau_{A+\kappa(\theta)}} \geq \nu, H_l \} = \{ \tilde{Z}_{\tau_{A+\kappa(\theta)}} = (1,l) \}, \quad l = 1,2,
\end{align*}
whose probability can be computed by Lemma \ref{lemma_non-iid}(1). In addition, for $l = 1,2$, \begin{align*}
(T_A - \nu)^+ 1_{H_l}= (\tilde{N}_{\tau_{A+\kappa(\theta)} } + 1 - \nu)^+ 1_{H_l} &= \Big( \sum_{t \leq \tau_{A+\kappa(\theta)}: \Delta \tilde{N}_t \neq 0} 1_{\{ \tilde{Z}_{t-} = (1,l) \}} \Big) + 1_{\{ \tilde{Z}_{\tau_A+\kappa(\theta)} = (1,l) \}},
%&=  \sum_{t \leq \tau_{A+\kappa(\lambda)}: \Delta N_t \neq 0} 1_{\{ \tilde{Z}_{t-} \in E_1^\nu \}} + 1_{\{N_{\tau_{A+\kappa(\lambda)} } \geq \nu \}}.
\end{align*} 
whose expectation can be computed by Lemma \ref{lemma_non-iid}(1) and (2).
%\[
% \tilde{\ba}  := [ \ba^{(0,1)} (1-\mu), \ba^{(1,1)} \mu (1-\epsilon),  \ba^{(1,m_1)} \mu \epsilon]
%\]
%
%In Example \ref{example_geometric}, 
%\[\e N_{\tau_x}=[(1-\mu)\bs \alpha_0, \mu \bs \alpha_1 ]\Big(\bs I_{2n}-\overline {\bs W}(x)(\tilde{\bs{T}}+ \bs B)\Big)^{-1}\overline {\bs W}(x)\begin{pmatrix}
%\bs t_0 \\ \bs t_1
%\end{pmatrix}.\]
%For the detection delay (counting after the change point), considering that the $+1$ at the end should be done only when it is in the post-change regime at the up-crossing time, we have
%\[[(1-\mu)\bs \alpha_0, \mu \bs \alpha_1 ]\Big(\bs I_{2n}-\overline {\bs W}(x)(\tilde{\bs{T}}+ \bs B)\Big)^{-1} \Big( \overline {\bs W}(x)\begin{pmatrix}
%\bs 0 \\ \bs t_1
%\end{pmatrix} + \begin{pmatrix}
%\bs 0 \\ \bs 1
%\end{pmatrix} \Big). \]
%False alarm is the event $\{\nu > T\}$ that it is stopped before the change. Its probability is an important quantity and it is given by
%\[[(1-\mu)\bs \alpha_0, \mu \bs \alpha_1 ]\Big(\bs I_{2n}-\overline {\bs W}(x)(\tilde{\bs{T}}+ \bs B)\Big)^{-1}
% \begin{pmatrix}
%\bs 1 \\ \bs 0
%\end{pmatrix}. \]

\end{remark}

\subsubsection{For the case $\theta < 0$}  Let $\bs W^d$ be the scale matrix of the MAP $(X^d, (\tilde{J},\tilde{Z}))$ as in \eqref{eq:MAP_general} with $\gamma = -\theta > 0$ and  $C \equiv -\kappa(\theta) > 0$ (see Remark \ref{remark_SA_CUSUM}).

%Recalling Proposition \ref{proposition_equivalence}(2), we
We first obtain the average run length and average detection delay, which can be derived easily by
Lemma \ref{lemma_SP_non_iid}. 
\begin{corollary}  \label{cor_CUSUM_noniid_SP} 
%Suppose $(T_A, \mathrm{P})  = (N_{A+ |\kappa(\theta)|},  \p_{|\kappa(\theta)|}^{\tilde{\ba}})$.
Fix $A > 0$. We have
\begin{align*}
\mathrm{ARL} (T_A ) &= \e_{ |\kappa(\theta)|} \big(\tilde{N}_{\tau_{A+ |\kappa(\theta)|}} |  (\tilde{J}_0,\tilde{Z}_0 ) \sim \tilde{\ba}  \big) \\
&= -   \tilde{\ba} \Big( \overline{\bs W}^d(A)- \bs W^d(A) (\bs W^{d})'_+(A+ |\kappa(\theta)|)^{-1} {\bs W}^d(A+ |\kappa(\theta)|)\Big) \tilde{\bt}, \\
\mathrm{ADD} (T_A ) &= \e_{|\kappa(\theta)|} \big((\tilde{N}_{\tau_{A+ |\kappa(\theta)|}} - \nu)^+ |  (\tilde{J}_0,\tilde{Z}_0 ) \sim \tilde{\ba}  \big) \\
&= -   \tilde{\ba} \Big( \overline{\bs W}^d(A)- \bs W^d(A) (\bs W^{d})'_+(A+ |\kappa(\theta)|)^{-1} {\bs W}^d(A+ |\kappa(\theta)|)\Big) \tilde{\bt}^{(1)}. 
\end{align*}
\end{corollary}
\begin{proof}
The first claim is immediate by Proposition \ref{proposition_equivalence}(2) and Lemma \ref{lemma_SP_non_iid}.
Regarding the second claim, because given $X_0 =  |\kappa(\theta)|$,
$(T_A - \nu)^+ = (\tilde{N}_{\tau_{A+ |\kappa(\theta)|}} - \nu)^+ =  \sum_{t \leq \tau_{A+|\kappa(\theta)|}: \Delta \tilde{N}_t \neq 0} 1_{\{ \tilde{Z}_{t-} \in E_1^\nu \}}$,
%&=  \sum_{t \leq \tau_{A+\kappa(\lambda)}: \Delta N_t \neq 0} 1_{\{ \tilde{Z}_{t-} \in E_1^\nu \}} + 1_{\{N_{\tau_{A+\kappa(\lambda)} } \geq \nu \}}.
Lemma \ref{lemma_SP_non_iid} shows the result.
%we have
%\begin{align*}
%\mathrm{E}(T_A - \nu)^+ 
%&= \sum_{i,k} \tilde{\ba}(i)  \e_{|\kappa(\theta)|,i}(N_{\tau_{A+ |\kappa(\theta)|}}(k)) 1_{\{ k \in E_1^\nu \}}.
%%  \\
%%&= -   \tilde{\ba} \Big( \overline{\bs W}^d(x-y)- \bs W^d(x-y) (\bs W^{d})'_+(x)^{-1} {\bs W}^d(x)\Big) \tilde{\bt}^{(1)}. 
%\end{align*}
%Now the result is immediate by Lemma  \ref{lemma_SP_non_iid}.
\end{proof}

On the other hand, the computation of the false alarm probability is more involved. By Proposition \ref{proposition_equivalence}(2), given $X_0 = |\kappa(\theta)|$,
$\{ T_A \leq \nu \} = \{ \tilde{N}_{\tau_{A+ |\kappa(\theta)|}} \leq \nu \}  = \{ \tilde{N}_{\tau_{A+ |\kappa(\theta)|}} < \nu \} \cup \{ \tilde{N}_{\tau_{A+ |\kappa(\theta)|}} = \nu \}$.
Here  \eqref{FA_identity} holds for ($\widetilde{\textrm{SP}}$)  as well and the probability of
%again
%
% noting that $\tilde{Z}$ enters $E_1^\nu$ at the arrival time $\tau^{(|\kappa(\lambda)|)}$ and hence,
$\{ \tilde{N}_{\tau_{A+ |\kappa(\theta)|}} < \nu \} = \{ \tilde{Z}_{\tau_{A+ |\kappa(\theta)|}} \in E_0^\nu \}$ can be computed by \eqref{first_passage_SP_noniid}. On the other hand, it is not clear if the probability of $\{ \tilde{N}_{\tau_{A+ |\kappa(\theta)|}} = \nu \} = \{ \tilde{Z}_{\tau_{A+ |\kappa(\theta)|}-} \in E_0^\nu, \tilde{Z}_{\tau_{A+ |\kappa(\theta)|}} \in E_1^\nu  \}$
can be directly computed.

However, this can be dealt by considering a modification, say $\widehat{Z}$, of $Z$ by doubling the states $E^\nu_1$ to keep track of whether it first entered from $E^\nu_0$ or not.
More precisely, we modify the state space of $Z$ to $E^\nu_0 \cup E^\nu_{1'} \cup E^\nu_1$ where $E^\nu_{1'}$ is a copy of $E^\nu_1$. The Markov chain  $\widehat{Z}$ moves from $E^\nu_0$ to $E^\nu_{1'}$ and then to $E^\nu_1$. In particular, it stays only at a unit time in $E^\nu_{1'}$. The change point is given by $ \nu = \inf (n \geq 0:  \widehat{Z}_n \notin E^\nu_0) = \inf (n \geq 0: \widehat{Z}_n \in E^\nu_{1'} \cup  E^\nu_1)$. 
%$0 \to 1' \to 1$ ($1'$ is for ``immediately after the change'').  
As a modification of \eqref{transition_discrete}, the transition matrix and initial distribution of $\widehat{Z}$ are given by, respectively,
\begin{align} \label{transition_discrete_SP_modified}
\begin{pmatrix}
\bs {K} & \bs {L} & \bs {O} \\ \bs {O} &  \bs{O} & \bs {M}  \\ \bs {O} &  \bs{O} & \bs {M} 
\end{pmatrix} \quad \textrm{and} \quad \bs{\beta} = (\bs \beta_1^{(0)}, \ldots, \bs \beta_{m_0}^{(0)}, \bs \beta_{1}^{(1)},  \ldots, \bs \beta^{(1)}_{m_1}, 0, \ldots, 0).
\end{align}
By replacing \eqref{transition_discrete} with \eqref{transition_discrete_SP_modified},  $\tilde{E}^\nu_0\cup \tilde{E}^\nu_1$, $\tilde{\ba}$, $\tilde{\mathfrak{n}}$, $\tilde{\bs T}$ and $\tilde{\bs B}$ are modified accordingly to say,   $\tilde{E}^\nu_0 \cup \tilde{E}^\nu_{1'}\cup \tilde{E}^\nu_1$,  $\widehat{\ba}$, $\widehat{\mathfrak{n}}$, $\widehat{\bs T}$ and $\widehat{\bs B}$.
%The process $\tilde{J}$ is modified to say $\widehat{J}$ in an obvious way. TODO
However,  these changes do not alter the law of  $\tilde{J}$ nor the arrivals $\tilde{N}$.
Following the same steps, we can compute \eqref{first_passage_SP_noniid} for this slightly generalized case. Indeed, given $X_0 = |\kappa(\theta)|$, 
\begin{multline*}
\{ T_A \leq \nu \} = \{ \tilde{N}_{\tau_{A+ |\kappa(\theta)|}} \leq \nu \}  = \{ \tilde{N}_{\tau_{A+ |\kappa(\theta)|}} < \nu \} \cup \{ \tilde{N}_{\tau_{A+ |\kappa(\theta)|}} = \nu \} \\
=  \{ \widehat{Z}_{\tau_{A+ |\kappa(\theta)|}} \in E_0^\nu \} \cup  \{ \widehat{Z}_{\tau_{A+ |\kappa(\theta)|}} \in E_{1'}^\nu  \} =  \{ \widehat{Z}_{\tau_{A+ |\kappa(\theta)|}} \in E_0^\nu \cup E_{1'}^\nu  \}.
\end{multline*}
By these and   \eqref{first_passage_SP_noniid}, the following is immediate.
\begin{corollary} \label{cor_CUSUM_noniid_SP2} 
%Suppose $(T_A, \mathrm{P})  = (N_{A+ |\kappa(\theta)|},  \p_{|\kappa(\theta)|}^{\tilde{\ba}})$. 
Fix $A > 0$.
Let $\bs W^d$ be the scale function of the MAP $(X^d, (\tilde{J},\widehat{Z}))$ as in \eqref{eq:MAP_general}  with $\gamma = -\theta > 0$ and  $C_i \equiv -\kappa(\theta) > 0$ and $\widehat{Z}$ given by \eqref{transition_discrete_SP_modified}. We have
%\begin{align*}
%\mathrm{PFA} ( T_A  ) = 
%\tilde{\ba} \p_{|\kappa(\theta)|} [(\tilde{J},\tilde{Z})_{\tau_{A+ |\kappa(\theta)|}}] \tilde{\bt}^{(0,1')}
%\end{align*}
\begin{align*}
\mathrm{PFA} ( T_A  )  &= \mathbb{P}_{|\kappa(\theta)|} \Big(\tilde{N}_{\tau_{A+|\kappa(\theta)|}}  \leq \nu | (\tilde{J}_0,\widehat{Z}_0) \sim \widehat{\ba} \Big) \\
&= \widehat{\ba}  \Big[ \bs I_{\widehat{n}} - \Big( \overline{\bs W}^d(A)  - \bs W^d(A)
(\bs W^{d})'_+(A+ |\kappa(\theta)|)^{-1}
{\bs W}^d(A+ |\kappa(\theta)|) \Big) (\widehat{\bs T}+\widehat{\bs B}) \Big] \widehat{\bs 1},
\end{align*}
where $\widehat{\bs 1} := [\bs 1^\top,  \ldots, \bs  1^\top, \bs 1^\top,  \ldots, \bs  1^\top, \bs  0^\top, \ldots, \bs  0^\top ]^\top$, whose element is $1$ if it belongs to $\tilde{E}_0^\nu \cup \tilde{E}_{1'}^\nu$ and zero otherwise.
\end{corollary}

\section{Numerical examples} \label{sec_numerics}

We conclude the paper by confirming the analytical results obtained in the previous sections through numerical experiments. All codes are implemented in Python.  Because the scale matrix grows exponentially fast, high precision is required for accurate results. Hence, we used the mpmath library with 30 digits. In addition, we compute the matrix exponentials in the scale function in an alternative way as described in Appendix \ref{remark_matrix_exp}.
% Windows 10 Pro, Intel(R) Xeon (R) E-2174G CPU @ 3.80GHz, 32.0GB of RAM. 
 We also used butools\footnote{Available at http://webspn.hit.bme.hu/$\sim$telek/tools/butools/doc/ph.html}  for randomly selecting the PH distributions used in our experiments.

We let $F_0$ be a PH distribution with initial distribution and transition rate matrix, respectively,
\[
  \bs  \alpha = (0.28,0.35,0.37), \quad \bs T = \begin{pmatrix}-0.51 & 0.12 & 0.12 \\ 0.21 & -0.46 & 0.10 \\ 0.28 & 0.16 & -0.63 \end{pmatrix}
    \]
    and $F_1$  be the PH distribution obtained by the exponential tilting of $F_0$ for \textbf{Case SN}: $\theta = 0.1$ and \textbf{Case SP}: $\theta = -0.1$.  As in Remark \ref{remark_sign}, the corresponding continuous-time process $X$ defined in \eqref{Sparre-andersen_CUSUM} becomes spectrally negative and spectrally positive, respectively.  Below, we focus on the LLR process \eqref{LLR} with $f_0$ and $f_1$ being the densities of $F_0$ and $F_1$, respectively.

For the barrier $A$, 
we set it to be the optimal barrier $A_\beta$ in the minimax formulation that minimizes the Lorden detection measure  \eqref{lorden_measure}
for $\beta = 5$ and $\beta = 10$, 
so that the average run length $\mathrm{ARL}(T_{A_\beta})=\mathrm{E}_\infty (T_{A_\beta})$ equals $\beta$.  
As discussed in Section  \ref{subsection_SA_CUSUM}, 
%Under $\mathrm{P}_\infty$, the LLR process $L$ is a random walk and $X$ becomes the Sparre-Andersen process as defined in \eqref{eq:MAP}.
 $\mathrm{ARL}(T_A)$, for any $A > 0$, is computed via 
 \eqref{ARL_SN} and  \eqref{ARL_SP}
%Lemmas  \ref{lemma_expectation} and \ref{lemma_expectation_dual} 
for \textbf{Case SN} and \textbf{Case SP}, respectively, using the scale matrix given in Corollary \ref{example_const}.  Because $A \mapsto \mathrm{ARL}(T_A)$ is monotonically increasing, we apply a classical bisection method with error bound $|\mathrm{ARL}- \beta| < 10^{-4}$. We obtain $A_5 = 0.456177$ and $A_{10} = 1.06076$ for \textbf{Case SN} and $A_5 =  0.994354$ and $A_{10} = 1.92654$ for \textbf{Case SP}.

%\jev{We can also mention the paper 'Computing the exponential of large block-triangular block-Toeplitz matrices encountered in fluid queues' exploring 2 approaches to a more general numerical problem.}

\subsection{Example 1: Geometric case and robustness}
We first consider a simple example with the  zero-modified geometric distributed change point as in Example \ref{example_robust}. We consider both the case $\epsilon= 0$ where the post-change distribution is certain to be $F_1$ and the case $\epsilon > 0$ where the  post-change distribution is a composite of $F_1$ and $F_2$, where we define $F_2$ to be another PH distribution given by
    \[
\bs \alpha =  (0.20, 0.25, 0.02, 0.18, 0.35) \quad \textrm{and} \quad \bs T= 
\begin{pmatrix} -1.45 & 0.35&0.34&0.34&0.05 \\
                0.01&-1.25&0.34&0.34&0.23 \\
                0.25 & 0.29&-0.70&0.10&0.02\\
                0.06&0.25&0.28&-1.01&0.16\\
                0.27&0.12&0.08&0.21&-0.87\end{pmatrix}.
\]

% for the non-robust case ($\varepsilon = 0$) and the robust cases $\varepsilon > 0$. 
%Example \ref{example_robust} with $\nu$ zero-modified geometric random variable TODO where $\mu = TODO$ and $\lambda = 0$ for the non-robust case ($\varepsilon = 0$) and the robust cases $\varepsilon > 0$. 
The corresponding scale matrix for the generalized Sparre-Andersen process is given in Theorem \ref{theorem_scale_matrix_general} (and we use Appendix \ref{remark_matrix_exp}), where   $\tilde{\ba}$, $\tilde{\bs{T}}$, $\tilde{\bs B} $ (of dimension $3+3+5 = 11$) are defined as in Example \ref{example_geometric}. Using this,  the average run length, average detection delay and false alarm probability are computed via Corollary  \ref{cor_CUSUM_noniid} for \textbf{Case SN} and Corollaries \ref{cor_CUSUM_noniid_SP}  and \ref{cor_CUSUM_noniid_SP2}  for  \textbf{Case SP}. 
 In order to confirm the accuracy of the obtained results, we compare them against those approximated by Monte Carlo simulation based on 100,000 sample paths. The results are summarized in Table \ref{Geometric case} for  \textbf{Cases SN} and \textbf{SP} and for $\beta = 5, 10$. Notice that the false alarm probability is invariant to the selection of $\epsilon$, because on $\{ T_A \leq \nu \}$, observations until $T_A$ are all $F_0$-distributed and does not depend on the post-change distribution.

\begin{table}[hbtp]
  \centering
    \begin{tabular}{c|c|cc|cc}
    \hline
     \multirow{2}{*}{$\epsilon$}
    & & \multicolumn{2}{|c}{$\beta = 5$}   & \multicolumn{2}{|c}{$\beta = 10$}  \\
     & & scale matrix &  simulation  & scale matrix &  simulation  \\
    \hline \hline
   \multirow{4}{*}{$0$} & ARL &  7.99071  & 7.99558 [7.93872, 8.05244] &   27.1271 & 26.8662 [26.6579, 27.0746] \\
   &ADD & 5.45165  &  5.45436 [5.39751, 5.51121] & 23.7523&23.4872 [23.2749,  23.6995] \\
    &PFA  &  0.49024  &  0.49163 [0.48850,  0.49476] & 0.28132 &0.28186 [0.27899, 0.28473] \\
    %&run time (in second)  &  .20538 &  428.47 & & \\
    \hline
       \multirow{4}{*}{$0.1$} & ARL& 7.70221 & 7.70868 [7.64808, 7.76928]&25.4470 & 25.3680 [25.1705, 25.5655]\\
   &ADD &5.16316 &  5.16996 [5.11107, 5.22885] & 22.0723 & 21.9988 [21.8009, 22.1967] \\
    &PFA  & 0.49024  &  0.48973 [0.48634, 0.49312] &  0.28132& 0.28235 [0.27944, 0.28526]\\
    %&run time (in second)  &  .20538 &  428.47 & & \\
    \hline
       \multirow{4}{*}{$0.5$} & ARL & 6.54824  & 6.52785 [6.48329, 6.57241] & 18.7267 &18.7092 [18.5514, 18.8670] \\
   &ADD  &  4.00919 & 3.99019 [3.94637, 4.03401] &15.3520 &  15.3208 [15.1653, 15.4763]\\
    &PFA  & 0.49024  & 0.49211 [0.48876, 0.49546] & 0.28132&  0.28155 [0.27878, 0.28432] \\
    %&run time (in second)  &  .20538 &  428.47 & & \\
    \hline
  \end{tabular} \\
  \textbf{Case SN} \\ \vspace{0.3cm}
    \begin{tabular}{c|c|cc|cc}
    \hline
     \multirow{2}{*}{$\epsilon$}
    & & \multicolumn{2}{|c}{$\beta = 5$}   & \multicolumn{2}{|c}{$\beta = 10$}  \\
     & & scale matrix &  simulation  &scale matrix &  simulation  \\
    \hline \hline
   \multirow{4}{*}{$0$} & ARL & 8.77778 & 8.74405  [8.68265, 8.80545] &  37.2390 & 37.3125 [37.0428, 37.5822]\\
   &ADD  &  5.99044  & 5.94159  [5.87837, 6.00481] &  33.4780 &33.5564 [33.2825, 33.8303]  \\
    &PFA & 0.42817  & 0.42919  [0.42601, 0.43237] & 0.18476 & 0.18428 [0.18191, 0.18665] \\
    %&run time (in second)  &  .20538 &  428.47 & & \\
    \hline
       \multirow{4}{*}{$0.1$} & ARL & 8.38546& 8.38532 [8.33715, 8.43349] & 34.4696 &34.4698 [34.2040, 34.7356] \\
   &ADD& 5.59812  & 5.59157 [5.54133, 5.64181] &  30.7086 & 30.6997 [30.4289, 30.9706]\\
    &PFA &  0.42817 & 0.42702 [0.42401, 0.43003] &  0.18476 & 0.18402 [0.18155, 0.18649]\\
    %&run time (in second)  &  .20538 &  428.47 & & \\
    \hline
       \multirow{4}{*}{$0.5$} & ARL & 6.81619 & 6.80107 [6.76191, 6.84023] &23.3923 & 23.2525 [23.0712, 23.4331]\\
   &ADD & 4.02885  &  4.01879 [3.97798, 4.05960] &  19.6313&19.4807 [19.2951, 19.6663]\\
    &PFA & 0.42817 & 0.42727 [0.42409, 0.43045]& 0.18476 &  0.18761 [0.18532, 0.18990]\\
    %&run time (in second)  &  .20538 &  428.47 & & \\
    \hline
  \end{tabular} \\
  \textbf{Case SP}  \vspace{0.3cm}
    \caption{\small (Example 1) The average run length (ARL), average detection delay (ADD), and false alarm probability (PFA) computed via the scale matrix and by Monte Carlo simulation (average and 95\% confidence interval) for the barrier $A_\beta$ for $\beta = 5, 10$ and $\epsilon = 0,0.1,0.5$ for \textbf{Case SN} (top) and \textbf{Case SP} (bottom).
    }
  \label{Geometric case}
\end{table}
  \begin{table}
    \centering
    \begin{tabular}{c|cc|cc}
    \hline
     & \multicolumn{2}{|c}{$\beta = 5$}   & \multicolumn{2}{|c}{$\beta = 10$}  \\
      & scale matrix &  simulation  & scale matrix&  simulation  \\
    \hline \hline
   ARL & 8.52856  &8.51199 [8.46315, 8.56083] & 24.8331 & 24.8246 [24.6512, 24.9980] \\
   ADD  &6.48684  & 6.46128 [6.41219, 6.51037] &   22.4024 &22.3787 [22.2044, 22.5529] \\
  PFA  &  0.30020  &0.30072 [0.29804, 0.30340] & 0.14769 &  0.14834 [0.14613, 0.15055] \\
    %&run time (in second)  &  .20538 &  428.47 & & \\
    \hline
  \end{tabular}  \\
  \textbf{Case SN} \\ \vspace{0.3cm}
    \begin{tabular}{c|cc|cc}
    \hline
     & \multicolumn{2}{|c}{$\beta = 5$}   & \multicolumn{2}{|c}{$\beta = 10$}  \\
      & scale matrix &  simulation  & scale matrix &  simulation  \\
    \hline \hline
   ARL &  6.71767  & 6.70485 [6.64896, 6.76074]&   24.2925 & 24.0729 [23.8398, 24.3060] \\
   ADD  &4.75621  & 4.74120 [4.68268, 4.79972]  &21.6381 & 21.5630 [21.3278, 21.7982] \\
    PFA & 0.33267  & 0.33309 [0.32996, 0.33622]&0.12034 & 0.12215 [0.12020, 0.12410] \\
    %&run time (in second)  &  .20538 &  428.47 & & \\
    \hline
  \end{tabular} 
  \\
  \textbf{Case SP} \vspace{0.3cm}
    \caption{\small (Example 2) The average run length (ARL), average detection delay (ADD), and false alarm probability (PFA) computed via the scale matrix and by Monte Carlo simulation (average and 95\% confidence interval) for the barrier $A_\beta$ for $\beta = 5, \beta = 10$ for \textbf{Case SN} (top) and \textbf{Case SP} (bottom).
    }
  \label{complicated_case}
\end{table}

\subsection{Example 2: More complex case} \label{subsection_example2} In order to confirm the accuracy and efficiency of the proposed method in more complex cases, we consider the following parameter set with extra heterogeneity of the observation distribution. We define $F_3$, $F_4$, $F_5$ obtained by the exponential tilting of $F_2$ with $\theta = 0.1,0.2,-0.05$, respectively.

For the Markov chain $Z$ of \eqref{transition_discrete}, we set $m_0 = 5$ and $m_1 = 3$ with 
\begin{align*}
\bs{\beta} &= (0.344, 0.312, 0.064, 0.056, 0.024, 0.06, 0.04, 0.100), \\
\bs {K} &=
\begin{pmatrix}
0.232  & 0.128 & 0.112 & 0.144 &  0.080 \\
0.080 & 0.352 & 0.112 & 0.112 & 0.056  \\
0.096  &  0.200 & 0.248 & 0.144  &0.016 \\   
0.048 & 0.072 & 0.064 &  0.480 & 0.056 \\  
0.128 &  0.120 & 0.056 & 0.024 & 0.448 
 \end{pmatrix}, \quad
 \bs {L} =
 \begin{pmatrix}
  0.304 &   0.000  &  0.000 \\
  0.288  &  0.000  &  0.000\\
    0.000 & 0.296  &  0.000 \\
 0.000  & 0.280  &  0.000 \\
  0.000  &  0.000 & 0.224
  \end{pmatrix}, \quad
  \bs {M} =
   \begin{pmatrix}
  1.0  &    0.0  &    0.0     \\
  0.0    &  0.3    & 0.7    \\
  0.0    &  0.5    & 0.5    
   \end{pmatrix}. 
\end{align*}
Notice that there are two absorbing classes $\{(1,1)\}$ and $\{(1,2), (1,3)\}$. The distributions of observation at pre-change states $(0,1), \ldots, (0,5)$ are set to be $F_0$, $F_3$, $F_4$, $F_5$, $F_0$, respectively, and those at post-change states $(1,1), (1,2), (1,3)$ are set to be  $F_1$, $F_2$, $F_3$, respectively. Different from Example 1, the absorption probability to each set depends on the underlying state on $E_0^\nu$. In addition, even after the change point, the observation fails to be stationary.

The corresponding scale matrix for the generalized Sparre-Andersen process is again given by Theorem \ref{theorem_scale_matrix_general} (again we use Appendix \ref{remark_matrix_exp}), but  $\tilde{\ba}$, $\tilde{\bs{T}}$, $\tilde{\bs B} $ this time is of dimension $(3+5+5+5+3)+(3+5+5)=34$. Besides the difference that the involved matrix is of higher dimensional, the algorithm remains the same.  Again, we compute  the performance measures via Corollary  \ref{cor_CUSUM_noniid} for \textbf{Cases SN} and Corollaries \ref{cor_CUSUM_noniid_SP}  and \ref{cor_CUSUM_noniid_SP2}  for  \textbf{Case SP}.  In Table \ref{complicated_case}, we compare the obtained results 
 against those approximated by Monte Carlo simulation based on 100,000 paths. %\jev{Kazu, thanks for all your work!! How do the running times compare? Is it worth the effort? :)}
 %Due to the involved matrix dimension, the run time is higher than in the previous example, but is substantially better than  simulation. Notice that these results are exact and stable, confirming the strength of this scale function approach.
% 
% \appendix
%\section{Remaining proofs}
%Here we establish a technical result required by the proof of Theorem~\ref{thm:W}
%assuming without real loss of generality that $\gamma=1$.
%\begin{lemma}\label{lem:analytic}
%In the defective case (killed process) the entries of 
%\[\sum_{k\geq 1}\bs e^\top_{k1}\e \big(e^{\bs T_k S_{k-1}} \big)e^{\bs T_kx}\bs e_{kk},\qquad x\in\mathbb C\] are entire functions. 
%\end{lemma}
%\begin{proof}
%As in~\eqref{eq:normbound} we see that the maximal absolute entry of $\bs T_ke^{\bs T_k x}$ is upper bounded by $re^{r|x|}$ for all $x\in \mathbb C,k\geq 1$.
%By noting that $\e e^{\bs T_k S_{k-1}}$ has $[0,1]$ entries we get a bound
%\[\sum_{k\geq 1}\left\|\bs e^\top_{k1}\e \big(e^{\bs T_k S_{k-1}}\big)\bs T_k e^{\bs T_kx}\bs e_{kk}\right\|\leq re^{r|x|}\sum_{k\geq 1}\left\|\bs e^\top_{k1}\e \big(e^{\bs T_k S_{k-1}}\big)\bs e_{kk}\right\|
%\leq re^{r|x|} \sum_{ij}\bs H(0)_{ij},\]
%where in the latter step we upper bounded the matrix norm by the sum over all non-negative entries. 
%In the defective case the matrix $\bs H(0)$ has finite entries~\cite{IP12}.
%According to, e.g.,~\cite[A16]{williams} differentiation at any $x\in\mathbb C$ can be performed under the summation sign and the result follows.
%\end{proof}

\appendix 

\section{Computation of the matrix exponential}  \label{remark_matrix_exp}

%\begin{remark}

In Corollary \ref{example_const} and Theorem \ref{theorem_scale_matrix_general},
%\[\bs W(x)=\frac{1}{\gamma}\sum_{k\geq 1}\int_{y\in[0,x]} \bs e^\top_{k1}e^{\bs T_k (y-x)/\gamma}\bs e_{kk} \p(S_{k-1}\in \D y), \qquad x\geq 0. \]
the  matrices $\bs T_k$ and $\tilde{\bs T}_k$ (especially the latter) tend to become large and the computation of the matrix exponentials can become unstable. Although Python and other programming languages (such as MATLAB) have built-in functions for approximating matrix exponentials, 
 we compute them in a different way by taking advantage of the form of  $\bs T_k$ and $\tilde{\bs T}_k$.

Fix $k \geq 1$.  By decomposing $\bs T_k = \bs U_k+ \bs V_k$ where 
\begin{align*}
\bs U_k :=\begin{pmatrix}
\bs T& \bs O &\ldots&\bs O \\
 &\ddots&\ddots&\vdots\\
\vdots&  &\bs T &\bs O \\
\bs O&\ldots& &\bs T
\end{pmatrix} \quad \textrm{and} \quad 
\bs V_k :=\begin{pmatrix}
\bs O& \bs B &\ldots&\bs O \\
 &\ddots&\ddots&\vdots\\
\vdots&  &\bs O &\bs B\\
\bs O&\ldots& &\bs O
\end{pmatrix},
\end{align*}
%$\bs U$ is the first diagonal blocks above the main diagonal (consisting of blocks $\bs B$ and $\bs O$)  \jev{I am confused here}, 
we have
\[
e^{\bs T_k (y-x)/\gamma} = \sum_{n=0}^\infty \frac {(y-x)^n} {\gamma^n n!} (\bs U_k + \bs V_k)^n, \quad 0 \leq y \leq x.
\]

Note that, by multiplication by $\bs V_k$, the locations of non-zero blocks are shifted  diagonally up-right while these are invariant to multiplication by $\bs U_k$. Hence, when considering the polynomial expansion of $(\bs U_k + \bs V_k)^n$, each term is a (diagonal) translation of a certain diagonal block matrix, where in particular the top right corner block is non-zero if and only if $\bs V_k$ is multiplied exactly $k-1$ times.  Hence, for $\bs e^\top_{k1}e^{\bs T_k (y-x)/\gamma}\bs e_{kk}$, which is the top right corner block of $e^{\bs T_k (y-x)/\gamma}$, only the subset of the terms in the expansion of  $(\bs U_k + \bs V_k)^n$ in which $\bs V_k$ is multiplied $k-1$ times (there are ${n+k} \choose {n}$ such terms), say  $V(\bs T, \bs B, n,k)$, needs to be considered. Thus we can write
\[
\bs e^\top_{k1}e^{\bs T_k (y-x)/\gamma}\bs e_{kk} = \sum_{n=0}^\infty \frac {(y-x)^n} {\gamma^n n!}  V(\bs T, \bs B, n,k), \quad 0 \leq y \leq x.
\]
We need to compute $V(\bs T, \bs B, n,k)$ for a range of $n$ and $k$. For their efficient computation, it is straightforward to construct a table of $(n,k) \mapsto V(\bs T, \bs B, n,k)$ by induction on $n$ and $k$, using elementary combinatorics. 

In particular, for the computation for CUSUM in Corollary \ref{example_const}, 
\begin{align*}\bs W(x)&=\frac{1}{\gamma}\sum_{k=1}^{\lfloor x/c\rfloor+1}
 \sum_{n=0}^\infty \frac {(c(k-1)-x)^n} {\gamma^n n!}  V(\bs T, \bs B, n,k), \quad x \geq 0, \\
 \bs W_+'(x)&= - \frac{1}{\gamma}\sum_{k=1}^{\lfloor x/c\rfloor+1}
 \sum_{n=1}^\infty \frac {(c(k-1)-x)^{n-1}} {\gamma^n (n-1)!}  V(\bs T, \bs B, n,k), \quad x > 0, \\
 \overline{\bs W}(x)
&=- \frac{1}{\gamma}\sum_{k=1}^{\lfloor x/c\rfloor+1}
\sum_{n=0}^\infty \frac {(c(k-1)-x)^{n+1}} {\gamma^{n} (n+1)!}  V(\bs T, \bs B, n,k), \quad x \geq 0, \end{align*}
where the last equality holds by Fubini's theorem. The above identities hold when $\bs T$ and $\bs B$ are replaced with   $\tilde{\bs T}$ and $\tilde{\bs B}$.
In our numerical results, we truncate the infinite series at $100$. Computing such a matrix is a basic question in PH renewal theory. We refer the reader to \cite{bini2016computing, bladt2017matrix} for techniques for computing related matrix exponentials. However, it is out of scope of this paper to evaluate the methods for matrix exponential.
%\end{remark}

\section{Proofs}

\subsection{Proof of Lemma \ref{lemma_expectation}} \label{proof_lemma_expectation}
%We consider the killed version of $(X,J)$ as follows.
Fix $k \in E$ and $a > 0$ throughout this proof.
%We modify slightly the MAP by adding killing upon arrival as follows. 
For $z \in (0,1]$, we let $\p^z$ be the law of $(X,J)$ when $J$ (and hence $X$ as well) is killed upon arrival from state $k$ with probability $1-z$ (and survives with probability $z$).
% where the background process $J$ is a Markov chain which transits according to $\bs T$ but instead of jumping into absorbing state it restarts with probability $z\in (0,1)$ when transiting from $i$ to $j$ (which happens at rate $\bt_i \ba_j = (\bt \ba)_{ij}$), and the jump of size $-C$ is introduced in~$X$.
The respective matrix exponent is given by
\begin{align} \label{F_z}
\bs F_z(s):=\gamma s\bs I_{\mathfrak{n}} +\bs T+\e (e^{-s C_1}) \bs B^{[z]}, \quad s \geq 0,
%\qquad \e^z [e^{s X_t};J_t]=e^{\bs F_z(s)t}, \quad t \geq 0, 
\end{align}
where for $l,j \in E$
\begin{align} \label{t_alpha_z}
\bs B^{[z]}_{lj} := \left\{ \begin{array}{ll} z\bs B_{lj}, & l=k, \\  \bs B_{lj}, & l \neq k. \end{array} \right.
\end{align}

Let ${\bs W}_z$ be the corresponding scale matrix. In particular, ${\bs W}_1 = {\bs W}$ for the original non-defective MAP and as in the proof of \cite[Thm.\ 1]{IP12} it is known that $\bs W_z \xrightarrow{z \uparrow 1} \bs W$. 
By \eqref{J_at_passage} applied under $\p^z$ and letting $\bs A_z := 
%\overline{\bs W}_z(x) F^z(0) =
 \overline{\bs W}_z(a) (\bs T+ \bs B^{[z]})
%\int_0^x W_z(y)\D yF_z(0)
$,
%According to~\cite[Thm.\ 2]{IP12},
\[\e_{0,i}( z^{N_{\tau_a}(k)} ) = \big[ \p^z [J_{\tau_a}]\bs 1 \big]_i = \big[ (\bs I_{\mathfrak{n}}-  \bs A_z
%\overline{\bs W}_z(x) F_z(0)
)^{-1}\bs 1 \big]_i, \quad i \in E, \]
where the first identity follows from the fact that the process must survive at each arrival from phase $k$ until $\tau_a$, each of which is a Bernoulli trial with success probability $z$.
Differentiating in $z$ and letting $z\uparrow 1$ we obtain
\begin{align} \label{N_temp}
\e_{0,i} (N_{\tau_a}(k)) =\Big[ \Big(\bs I_{\mathfrak{n}} - \overline{\bs W}(a) (\bs T+ \bs B)\Big)^{-1}  \overline{\bs W}(a)\Big]_{i k}  \bt_{k}.
\end{align}
%\[\e N_{\tau_x}(i,j) = \Big[\ba \Big(I- \overline{\bs W}_1(x) (T+\bt\ba)\Big)^{-1}  \overline{\bs W}_1(x)\Big]_{i}  (\bt \ba)_{ij}
%%=\ba \Big(I- \overline{\bs W}_1(x) (T+\bt\ba)\Big)^{-1}  \overline{\bs W}_1(x)\bt. 
%\]
%where $W_1$ denotes the scale function of non-killed process, i.e., for $z=1$.

To see this, note that differentiability of $\overline{\bs W}_z(x)$ in $z$ follows from the identity~\eqref{J_at_passage} for the killed process,  and then
\[
 \frac \D {\D z} (\bs I_{\mathfrak{n}}- \bs A_z)^{-1} \bs 1 =   (\bs I_{\mathfrak{n}}-\bs A_z)^{-1} \Big( \frac \D {\D z} \bs A_z \Big)  (\bs I_{\mathfrak{n}} -\bs A_z)^{-1} \bs 1. %\xrightarrow{z \uparrow 1} \ba \Big(I- \overline{\bs W}_1(x) (T+\bt\ba)\Big)^{-1} \overline{\bs W}_1(x) I^{ij} (\bt \ba)_{ij} \bs 1. 
\]
We have $(\bs I_{\mathfrak{n}} - \bs A_z)^{-1} \bs 1 
%= (\bs I_n+ \bs A_z+ \bs A_z^2 + \cdots)  \bs 1 
=  \bs 1  + (\bs I_{\mathfrak{n}}- \bs A_z)^{-1} \bs A_z \bs 1 \xrightarrow{z \uparrow 1} \bs 1$  observing  that $ \bs A_z \bs 1 \xrightarrow{z \uparrow 1} \overline{\bs W}(a) (\bs T+\bs B)\bs 1=\bs 0$ (recall our assumption that the original MAP is non-defective and hence $(\bs T+\bs B)\bs 1=\bs 0$). In addition, because again  $(\bs T+\bs B) \bs 1 = \bs 0$,
\begin{multline*}
\frac \D {\D z} \bs A_z \bs 1 = 
%\int_0^a \frac \partial {\partial z} {\bs W}_z(y)\D y  
 \frac {\D \overline{\bs W}_z(a)} {\D z}
 (\bs T+ \bs B^{[z]}) \bs 1+ \overline{\bs W}_z(a)  \frac \D {\D z}  (\bs T+ \bs B^{[z]}) \bs 1 \xrightarrow{z \uparrow 1}  \overline{\bs W}(a) \frac \D {\D z}  (\bs T+ \bs B^{[z]}) |_{z=1} \bs 1. 
\end{multline*}
Because $\frac \D {\D z}  (\bs T+ \bs B^{[z]}) |_{z=1}$ is the matrix whose $(l,j)$-th entry is  $\bs B_{lj}$ when $l=k$ and zero otherwise,  
$\frac \D {\D z}  (\bs T+ \bs B^{[z]}) |_{z=1} {\bs 1}$ is the vector with $k$-th element $\sum_{j \in E}\bs B_{kj} = \sum_{j \in E}\bs t_{k} \ba_j  = \bt_k$ and zero for others. Hence, \eqref{N_temp}, equivalently the first claim of this lemma, holds. 
The second claim is a direct consequence of the first claim.
%\[\e N_{\tau_x}(k) = \Big[\ba \Big(I- \overline{\bs W}_1(x) (T+\bt\ba)\Big)^{-1}  \overline{\bs W}_1(x)\Big]_{k}  (\bt \ba)_{k}.
%\]
%
%\red{[to be deleted]}
%Here we used the identity $(I-A)^{-1}=I+A+A^2$ for $A=\int_0^x W(y)\D yF(0)$ and observed that $(T+\bt\ba)\bs 1=\bs 0$.
%\jev{Technical details can be added later}
%\kazu{Do you mean this? We have  $(I-A)^{-1} \bs 1 =\bs 1 +A \bs 1 +A^2 \bs 1 + \cdots$. We have $A {\bf 1} \xrightarrow{z \uparrow 1} \int_0^x W_1(y)\D y ( T+\bt\ba) = \bs 0$. Hence, $\frac \partial {\partial z} A^n \bs 1  \xrightarrow{z \uparrow 1} A^{n-1} |_{z=1} \frac \partial {\partial z} A \bs 1|_{z=1} = \Big(\int_0^x W_1(y)\D y(T+\bt\ba) \Big)^{n-1} \int_0^x W_1(y)\D y\bt$. Substituting this, we have
%\[
%\ba \frac \partial {\partial z}(I-A)^{-1} \bs 1 \xrightarrow{z \uparrow 1} \ba \Big(I-\int_0^x W_1(y)\D y(T+\bt\ba)\Big)^{-1}\int_0^x W_1(y)\D y\bt.
%\]
%}
%\jev{Yes, the two minus signs cancel out. We can also differentiate the matrix inverse directly.}

\subsection{Proof of Lemma \ref{lemma_expectation_dual} } \label{proof_lemma_expectation_dual}
Fix $k \in E$ and $a > 0$ throughout this proof.
Similar to the proof of Lemma \ref{lemma_expectation}, for $z \in (0,1]$, we let $\p^z$ be the law of the killed version with modification  \eqref{t_alpha_z} and ${\bs W}^d_z$ be the corresponding scale matrix of the spectrally negative MAP $(X^d, J)$.
%consider the killed version modeled by the matrix exponent \eqref{F_z} for fixed $k \in E$.
%We let $\bs F^d_z$ be the matrix exponent of $X^d$ under $\p^z$, which has exactly the same structure as in \eqref{F_z}. 
%
%\begin{align}
%\bs F(s):=\gamma s\bs I_n+\bs T+\e e^{-s C_1}\bt\ba ,\qquad \e [e^{s (X_t-X_0)};J_t ]  =e^{\bs F(s)t}, \quad s \geq 0,  \label{def_F}
%\end{align}
%Assuming $\theta<0$ we define the dual process $\tilde X=-X$ and the corresponding scale function~$\tilde W$. 
%Note that $\tilde F(s)$ has exactly the same structure as $F(s)$ for $\theta>0$.
By \eqref{first_passage_SP} applied under $\p^z$, for $i \in E$,
% because $(\bs T+\bt\ba) \bs 1 = \bs 0$
\begin{align*}\e_{x, i} (z^{N_{\tau_a}(k)}) &= [\p_{x}^z [J_{\tau_a}]\bs 1 ]_i = \Big( \Big[\bs I_{\mathfrak{n}} - \Big( \overline{\bs W}^d_z(a-x)  - \bs W^d_z(a-x)
(\bs W_z^{d})'_+(a)^{-1}
{\bs W}^d_z(a) \Big) \big(\bs T+ \bs B^{[z]} \big)
 \Big] \bs 1 \Big)_i \\
&=1+(1-z) \Big(\overline{\bs W}^d_z(a-x)    - {\bs W}^d_z(a-x)   (\bs W_z^{d})'_+(a)^{-1}
{\bs W}^d_z(a) \Big)_{ik}\bt_k.
\end{align*}
To see the latter equality, because $(\bs T+\bs B) \bs 1 = \bs 0$, the vector $ \big(\bs T+ \bs B^{[z]} \big)
\bs 1= (\bs B^{[z]} - \bs B ) \bs 1$ has its  $l$-th element equal to $(z-1) \sum_{j \in E}\bs B_{kj} = (z-1) \sum_{j \in E}\bs t_{k} \ba_j  = (z-1) \bt_k$ if $l=k$ and zero otherwise. 

%where $\tilde W$ depends on~$z$. 
The derivative of the right hand side of the above display at $z=1$ is given by
\[-\lim_{z\uparrow 1}\Big(\overline{\bs W}^d_{z}(a-x)    - {\bs W}^d_{z}(a-x)   (\bs W_{z}^{d})'_+(a)^{-1}
{\bs W}^d_{z}(a) \Big)_{ik}\bt_k,\] and it is left to note that the respective quantities converge,  see~\cite{IP12} (the proof of Thm.\ 1 and the identity in Thm.\ 5). The second claim holds immediately by the first claim.
%\begin{align*}\e_{y, i} (N_{\tau_x}(k))
%&=-\Big(\overline{\bs W}^d_1(x-y)    - {\bs W}^d_1(x-y)   (\bs W_1^{d})'_+(x)^{-1}
%{\bs W}^d_1(x) \Big)_{i, k}\bt_k.
%\end{align*}
%\[\e N_{\tau_x}(k)=-\ba \Big(\int_0^x \tilde W_1(y)\D y-\tilde W_1(x)\tilde W'_{1+}(x)^{-1}\tilde W_1(x)\Big)_{\cdot, k}\bt_k.\]

\bibliographystyle{abbrv}
\bibliography{changempointMAP}

\end{document}